\documentclass[11pt]{article} 

\usepackage[utf8]{inputenc} 

\usepackage{geometry} 
\usepackage{graphicx} 
\usepackage{color}
\usepackage{blkarray}

\definecolor{Red}{rgb}{1,0,0}
\definecolor{Blue}{rgb}{0,0,1}
\definecolor{Olive}{rgb}{0.41,0.55,0.13}
\definecolor{Green}{rgb}{0,1,0}
\definecolor{MGreen}{rgb}{0,0.8,0}
\definecolor{DGreen}{rgb}{0,0.55,0}
\definecolor{Yellow}{rgb}{1,1,0}
\definecolor{Cyan}{rgb}{0,1,1}
\definecolor{Magenta}{rgb}{1,0,1}
\definecolor{Orange}{rgb}{1,.5,0}
\definecolor{Violet}{rgb}{.5,0,.5}
\definecolor{Purple}{rgb}{.75,0,.25}
\definecolor{Brown}{rgb}{.75,.5,.25}
\definecolor{Grey}{rgb}{.5,.5,.5}

\usepackage{booktabs} 
\usepackage{array} 
\usepackage{paralist} 
\usepackage{verbatim} 
\usepackage{subfigure} 
\usepackage{amsmath,amssymb,amsthm,mathrsfs}
\usepackage[colorlinks]{hyperref}
\usepackage{blkarray}

\usepackage{fancyhdr} 
\makeatletter
\newtheorem*{rep@theorem}{\rep@title}
\newcommand{\newreptheorem}[2]{%
\newenvironment{rep#1}[1]{%
 \def\rep@title{#2 \ref{##1}}%
 \begin{rep@theorem}}%
 {\end{rep@theorem}}}
\makeatother

\theoremstyle{plain}

\newtheorem{theorem}{Theorem}[section] 
 
\newtheorem{corollary}{Corollary}[section]

\newtheorem{lemma}{Lemma}[section]

\newtheorem{theorem*}{Theorem}   
\newtheorem{lemma*}{Lemma} 
\newtheorem{corollary*}{Corollary} 
\newtheorem{remark*}{Remark}
\newtheorem*{example*}{Example}

\newtheorem{remark}{Remark}[section]

\newlength{\widebarargwidth}
\newlength{\widebarargheight}
\newlength{\widebarargdepth}

\theoremstyle{definition}

\def\cL{{\cal L}}
\def\cM{{\cal M}}
\def\cN{{\cal N}}

\def\cS{{\cal S}}

\def\cW{{\cal W}}

\newcommand{\mprob}{\ensuremath{\mathbb{P}}}
\newcommand{\argmax}{\ensuremath{\operatorname{argmax}}}

\begin{document}

\begin{center}

{\bf{\Large{Information-theoretic bounds for exact recovery in weighted stochastic block models using the Renyi divergence}}}

\vspace*{.25in}

\begin{tabular}{ccc}
{\large{Varun Jog}} & \hspace*{.75in} & {\large{Po-Ling Loh}} \\ 
{\large{\texttt{varunjog@wharton.upenn.edu}}} & & {\large{\texttt{loh@wharton.upenn.edu}}} \vspace{.2in}
 \\
Departments of Statistics \& CIS & \hspace{.2in} & Department of Statistics \\
Warren Center for Network and Data Sciences && The Wharton School \\
University of Pennsylvania && University of Pennsylvania \\ Philadelphia, PA 19104 & & Philadelphia, PA 19104
\end{tabular}

\vspace*{.2in}

September 2015

\vspace*{.2in}

\end{center}


\begin{abstract}
We derive sharp thresholds for exact recovery of communities in a weighted stochastic block model, where observations are collected in the form of a weighted adjacency matrix, and the weight of each edge is generated independently from a distribution determined by the community membership of its endpoints. Our main result, characterizing the precise boundary between success and failure of maximum likelihood estimation when edge weights are drawn from discrete distributions, involves the Renyi divergence of order $\frac{1}{2}$ between the distributions of within-community and between-community edges. When the Renyi divergence is above a certain threshold, meaning the edge distributions are sufficiently separated, maximum likelihood succeeds with probability tending to 1; when the Renyi divergence is below the threshold, maximum likelihood fails with probability bounded away from 0. In the language of graphical channels, the Renyi divergence pinpoints the information-theoretic capacity of discrete graphical channels with binary inputs. Our results generalize previously established thresholds derived specifically for unweighted block models, and support an important natural intuition relating the intrinsic hardness of community estimation to the problem of edge classification. Along the way, we establish a general relationship between the Renyi divergence and the probability of success of the maximum likelihood estimator for arbitrary edge weight distributions. Finally, we discuss consequences of our bounds for the related problems of censored block models and submatrix localization, which may be seen as special cases of the framework developed in our paper.
\end{abstract}

\section{Introduction}

The recent explosion of interest in network data has created a need for new statistical methods for analyzing network datasets and interpreting results~\cite{NewEtal06, DavKle10, Jac10, GolEtal10}. One active area of research with diverse applications in many scientific fields pertains to community detection and estimation, where the information available consists of the presence or absence of edges between nodes in the graph, and the goal is to partition the nodes into disjoint groups based on their relative connectivity~\cite{FieEtal85, HarSha00, PriEtal00, ShiMal00, McS01, NewGir04}.

A standard assumption in statistical modeling is that conditioned on the community labels of the nodes in the graph, edges are generated independently according to fixed distributions governing the connectivity of nodes within and between communities in the graph. This is the setting of the stochastic block model (SBM)~\cite{HolEtal83, HarEtal76, WasAnd87}. In the homogeneous case, edges follow one distribution when both endpoints are in the same community, regardless of the community label; and edges follow a second distribution when the endpoints are in different communities. A variety of interesting statistical results have been derived recently characterizing the regimes under which \emph{exact} or \emph{weak} recovery of community labels is possible (e.g., \cite{MosEtal12, MosEtal14, Mas14, AbbEtal14, abbe2014exact, AbbSan15, HajEtal14, HajEtal15, zhangminimax}). Exact recovery refers to the case where the communities are partitioned perfectly, and a corresponding estimator is called \emph{strongly consistent}. On the other hand, weak recovery refers to the case where the estimated community labels are positively correlated with the true labels.

In the setting of stochastic block models with nearly-equal community sizes and homogeneous connection probabilities, Zhang and Zhou~\cite{zhangminimax} derive minimax rates for statistical estimation in the case of exact recovery. Interestingly, the expression they obtain contains the Renyi divergence of order $\frac{1}{2}$ between two Bernoulli distributions, corresponding to the probability of generation for within-community and between-community edges. Hence, the hardness of recovering the community assignments is somehow captured in the hardness of inferring whether pairs of nodes lie within the same community or in different communities. This result has a very natural intuitive interpretation, since knowing whether each pair of nodes (or even each pair of nodes along the edges of a spanning tree of the graph) lies in the same community would clearly lead to perfect recovery of the community labels. On the other hand, this constitutes a somewhat different perspective from the prevailing viewpoint of the hardness of recovering community labels being innately tied to the success or failure of a hypothesis testing problem determining whether an individual node lies in one community or another \cite{AbbSan15, MosEtal14, zhangminimax}. Several other attempts have been made to relate the sharp threshold behavior of community estimation to various quantities in information theory~\cite{AbbMon13, CheEtal15, DesEtal15, AbbSan15}, but the precise relationship is still largely unknown.

The vast majority of existing literature on stochastic block models has focused on the case where no other information is available beyond the unweighted adjacency matrix. In an attempt to better understand the information-theoretic quantities at work in determining the thresholds for exact recovery in stochastic block models, we will widen our consideration to the more general weighted problem. Note that situations naturally arise where network datasets contain information about the strength or type of connectivity between edges, as well~\cite{New04, BocEtal06}. In social networks, information may be available quantifying the strength of a tie, such as the number of interactions between the individuals in a certain time period~\cite{Sad72}; in cellular networks, information may be available quantifying the frequency of communication between users~\cite{BloEtal08}; in airline networks, edges may be labeled according to the type of air traffic linking pairs of cities~\cite{BarEtal04}; and in neural networks, edge weights may symbolize the level of neural activity between regions in the brain~\cite{RubSpo10}. Of course, the connectivity data could be condensed into an adjacency matrix consisting of only zeros and ones, but this would result in a loss of valuable information that could be used to recover node communities.

In this paper, we analyze the ``weighted" setting of the stochastic block model, where edges are generated from arbitrary distributions that are not restricted to being Bernoulli. Our key question is whether the Renyi divergence of order $\frac{1}{2}$ appearing in the results of Zhang and Zhou~\cite{zhangminimax} continues to persist as a fundamental quantity that determines the hardness of exact recovery in the generalized setting. Surprisingly, our answer is affirmative. First, we show that the Renyi divergence between the within-community and between-community edge distributions may be used directly to control the probability of failure of the maximum likelihood estimator. Hence, as the Renyi divergence increases, corresponding to edge distributions that are further apart, the probability of failure of maximum likelihood is driven to zero. Next, we focus on a specific regime involving discrete weights (or colors), where the average number of edges of each specific color connected to a node scales according to $\Theta(\log n)$. In this case, we show that the bounds derived earlier involving the Renyi divergence are in fact tight, and exact recovery is impossible when the Renyi divergence between the weighted distributions is below a certain threshold. Our results are also applicable in the more general setting of more than two communities. Finally, we discuss the consequences of our theorems in the context of decoding in discrete graphical channels and submatrix localization with continuous distributions.

The remainder of the paper is organized as follows: In Section~\ref{SecBackground}, we introduce the basic background and mathematical notation used in the paper. In Section~\ref{SecMain}, we present our main theoretical contributions, beginning with achievability results for the maximum likelihood estimator in a weighted stochastic block model with arbitrarily many communities. We then derive sharp thresholds for exact recovery in the discrete weighted case, and then interpret our results in the framework of graphical channels and submatrix localization. Section~\ref{SecProofs} contains the main arguments for the proofs of our theorems. We conclude in Section~\ref{SecDiscussion} with a discussion of several open questions related to phase transitions in weighted stochastic block models.


\section{Background and problem setup}
\label{SecBackground}

Consider a stochastic block model with $K \ge 2$ communities, each with $n$ nodes. For each node $i$, let $\sigma(i) \in \{1, 2, \dots, K\}$ denote the community assignment of the node. A weighted stochastic block model consists of a random graph generated on the vertices $\{1, 2, \dots, nK\}$, using the community assignments $\sigma$, as well as a sequence of distributions $p_n^{(k_1, k_2)} (= p_n^{(k_2, k_1)})$, for $1 \leq k_1, k_2 \leq K$ and $n \geq 1$. The support of the distributions may be continuous or discrete. In the discrete case, we will often use the terms weight, color, and label interchangeably. The weighted random graph is generated as follows: Each edge $(i,j)$ is assigned a random weight $W_{(i,j)} \sim p^{(\sigma(i), \sigma(j))}_n$, independent of the weights of all other edges. Such a stochastic block model is called \emph{non-homogeneous}, since the distributions of the edge weights depend not only on whether the endpoints of an edge belong to the same community, but also on which communities they belong to.

In this paper, we will consider a \emph{homogeneous} weighted stochastic block model, which may be described simply as follows: Given a sequence of distributions $\{p_n\}$ and $\{q_n\}$, every edge $(i, j)$ is assigned a random weight $W_{(i,j)}$, independently of all other edge weights, such that
\begin{align}
\label{EqnHomogeneous}
W_{(i,j)} \sim 
\begin{cases}
p_n &\text{ if } \sigma(i) = \sigma(j),\\
q_n &\text{ if } \sigma(i) \neq \sigma(j).
\end{cases}
\end{align}
The traditional (unweighted) stochastic block models constitute a special case of weighted stochastic block models, since we may encode edges with weights 1 or 0, corresponding to the presence or absence of an edge.

Our ultimate goal is to infer the underlying communities based on observing the weight matrix $W$. Several differing notions of inference have been studied in the case of unweighted stochastic block models. In the ``sparse regime," where the distributions $p_n$ and $q_n$ scale as
\begin{align*}
p_n(0) = \frac{1-a/n}{n}, \qquad &p_n(1) = \frac{a}{n}, \quad \text{and}\\
q_n(0) = \frac{1-b/n}{n}, \qquad &q_n(1) = \frac{b}{n},
\end{align*}
for constants $a, b \geq 0$, one cannot hope to recover the communities exactly, since the graph is not connected with high probability. The notion of ``detection" or ``weak recovery" considered in this regime consists of obtaining community assignments that are positively correlated with the true assignment. It has been shown in the case $K = 2$ that if
\begin{equation}\label{eq: an1}
(a-b)^2 > a+b,
\end{equation}
it is impossible to obtain such an assignment\footnote{We appropriately modify the conditions  to take into account that the community size in our setting is $n$, as opposed to $n/2$.}; whereas if
\begin{equation*}
(a-b)^2 < a+b,
\end{equation*}
obtaining a positively correlated assignment becomes possible~\cite{MosEtal13, Mas14}. 

In order to obtain exact recovery, a simple necessary condition is that the graph must be connected, meaning the probability of having an edge must scale according to $\Omega\left(\frac{\log n}{n}\right)$. This regime was considered in Abbe et al.~\cite{abbe2014exact}, where the probabilities were given by
\begin{align*} 
p_n(0) = \frac{1-a\log n/n}{n}, \qquad &p_n(1) = \frac{a \log n}{n}, \quad \text{and}\\
q_n(0) = \frac{1-b\log n/n}{n}, \qquad &q_n(1) = \frac{b \log n}{n},
\end{align*}
for constants $a, b \geq 0$. In this regime, it was shown \cite{abbe2014exact} that exact recovery of communities is possible if 
\begin{equation*}
\left| \sqrt{a} - \sqrt{b}\right| > 1,
\end{equation*}
and impossible if 
\begin{equation*}
\left| \sqrt{a} - \sqrt{b}\right| < 1.
\end{equation*}
Apart from exact recovery (also known as strong consistency) and weak recovery, a notion of partial recovery (also known as weak consistency) has also been considered~\cite{MosEtal14, AminiEtal13, zhangminimax}. This notion lies between the other two notions of recovery, and only requires the fraction of misclassified nodes to converge in probability to 0 as $n$ becomes large. A very general result for the $K=2$ case, characterizing when exact and partial recovery are possible for the unweighted homogeneous stochastic block model, is provided in Mossel et al.~\cite{MosEtal14}. Zhang and Zhou~\cite{zhangminimax} consider the problem of community detection in a minimax setting with an appropriate loss function, where the parameter space consists of both homogeneous and non-homogeneous stochastic block models, the number of communities may be fixed or growing, and the community sizes need not be exactly equal. In particular, for the case of homogeneous stochastic block models where the community sizes are almost equal and scale as $\frac{n(1+o(1))}{K}$, they show that the loss function decays at the rate of $e^{-(1+o(1))nI/K}$ whenever $\frac{nI}{K} \to \infty$. Here, $I$ is the Renyi divergence of order $\frac{1}{2}$ between the two Bernoulli distributions corresponding to between-community and within-community edges. Furthermore, they show that exact recovery is possible if and only if the loss function is $o(n^{-1})$, whereas partial recovery is possible if and only if it is $o(1)$. The exact recovery bounds achieved in this way match those of Abbe et al.~\cite{abbe2014exact}.

Heimlicher et al.~\cite{HeiEtal12} also conjectured that similar threshold phenomena should exist in the case of the stochastic block model with discrete weights. In particular, Heimlicher et al.~\cite{HeiEtal12} consider the homogeneous case where $K = 2$ and the between-community and within-community connection probabilities scale as $\Theta\left(\frac{1}{n}\right)$. Analogous to  expression \eqref{eq: an1}, they conjectured a threshold in terms of the discrete probabilities such that weak recovery is possible above this threshold and impossible below the threshold. The impossibility of reconstruction below the conjectured threshold was established in Lelarge et al.~\cite{LelEtal13}, and efficient algorithms that achieve weak recovery were provided for a constant above the threshold.

In this paper, we consider the problem of exact recovery in the homogeneous weighted stochastic block model with $K \geq 2$ communities. By definition, the estimator that minimizes the probability of erroneous community assignments is the maximum likelihood estimator: If the maximum likelihood estimator fails to recover the communities with a certain probability, then the probability of error of any other estimator is also lower-bounded by the same probability. Thus, to show impossibility of recovery, it is sufficient to show that the maximum likelihood estimator fails with a nonzero probability. Finally, note that as in the unweighted case, the maximum likelihood estimator in the weighted case is easy to describe in terms of a min-cut graph partition~\cite{LelEtal13}. Let $\cL$ be the class of edge labels, and let $p_n$ and $q_n$ be distributions supported on $\cL$ which describe the probabilities of edge labels for within-community and between-community edges. For an edge with label $\ell \in \cL$, we assign a weight of $\log \left(\frac{p_n(\ell)}{q_n(\ell)} \right)$. The maximum likelihood estimator then seeks to partition the vertices into disjoint communities in such a way that the sum of weights of between-community edges is minimized.


\section{Main results and consequences}
\label{SecMain}

In this section, we present our main results concerning achievability and impossibility of exact recovery, along with several applications.

\subsection{Renyi divergence and achievability}
\label{SecAchievable}

We begin with a result that controls the probability of success for maximum likelihood estimation under the general homogeneous model~\eqref{EqnHomogeneous}, when $K = 2$. Our first theorem relates the probability of failure of maximum likelihood to the Renyi divergence between the distributions for within-community and between-community edge weights.

\begin{theorem}[Proof in Section~\ref{SecThmRenyiML}]
\label{ThmRenyiML}
Consider a stochastic block model with two communities of size $n$, with connection probabilities governed by the model~\eqref{EqnHomogeneous}. Then the probability that the maximum likelihood estimator fails is bounded as
\begin{equation}
\label{eq: mallu1}
\mprob(F) \le \sum_{k=1}^{n/2} \exp\left(2k \left(\log \frac{n}{k} + 1\right) - 2k (n-k) I\right),
\end{equation}
where $I$ is the Renyi divergence of order $\frac{1}{2}$ between the edge weight distributions $p_n(x)$ and $q_n(x)$, given by
\begin{equation*}
I =
\begin{cases}
-2 \log \left(\int_{-\infty}^\infty \sqrt{p_n(x) q_n(x)} dx\right), & \quad \text{for continuous distributions on } \mathbb R, \\
-2\log\sum_{\ell \geq 0} \sqrt{p_n(\ell) q_n(\ell)}, & \quad \text{for discrete distributions on } \mathbb N.
\end{cases}
\end{equation*} 
\end{theorem}

Note that the general exponential bound in inequality~\eqref{eq: mallu1} decreases with $I$, which corresponds to the distributions $p_n$ and $q_n$ becoming more separated. This corroborates the intuition that the failure probability of maximum likelihood $\mprob(F)$ appearing on the left-hand side of inequality~\eqref{eq: mallu1} should decrease with $I$, since the problem becomes easier to solve as the within-community and between-community distributions become easier to distinguish. \\

Of course, Theorem~\ref{ThmRenyiML} is particularly informative in regimes where we can show that the right-hand side of inequality~\eqref{eq: mallu1} tends to 0, implying that the maximum likelihood estimator succeeds with probability tending to 1. To illustrate this point, we have the following corollary:

\begin{corollary}[Proof in Section~\ref{SecCorRenyi}] 
\label{CorRenyi}
Suppose the Renyi divergence between $p_n$ and $q_n$ satisfies
\begin{equation*}
\liminf_{n \rightarrow \infty} \frac{nI}{\log n} > 1.
\end{equation*}
Then the maximum likelihood estimator succeeds with probability converging to 1 as $n \rightarrow \infty$.
\end{corollary}
We will discuss the implications of Corollary~\ref{CorRenyi} in various scenarios in the sections below. We also have a version of Theorem~\ref{ThmRenyiML} that is applicable to the case of more than two communities. We state and prove the more general theorem separately, since the argument for $K = 2$ is substantially simpler.

\begin{theorem}[Proof in Section \ref{SecThmRenyiK}]
\label{ThmRenyiK}
Consider a stochastic block model with $K$ communities of size $n$, with connection probabilities governed by the model~\eqref{EqnHomogeneous}. Then the probability that the maximum likelihood estimator fails is bounded as
\begin{equation}
\label{EqnDucky}
\mprob(F) \le \sum_{m=1}^{\lfloor n/2 \rfloor} \min \left\{ \left( \frac{enK^2}{m}\right)^m, \; K^{nK}\right\} e^{(-nm + m^2)I} + \sum_{m=\lfloor n/2 \rfloor + 1}^{nK} \min \left\{ \left( \frac{enK^2}{m}\right)^m, \; K^{nK}\right\} e^{-\frac{2mn}{9} I},
\end{equation}
where $I$ is the Renyi divergence of order $\frac{1}{2}$ between the edge weight distributions $p_n(x)$ and $q_n(x)$. In particular, if
\begin{equation}
\label{EqnPiggy}
\liminf_{n \rightarrow \infty} \frac{nI}{\log n} > 1,
\end{equation}
then the maximum likelihood estimator succeeds with probability converging to 1 as $n \rightarrow \infty$.
\end{theorem}

The proof of Theorem~\ref{ThmRenyiK} builds upon the arguments of Zhang and Zhou~\cite{zhangminimax} and extends them to more general distributions.



\subsection{Thresholds for weighted stochastic block models}
\label{SecDiscrete}

In this section, we derive a threshold phenomenon for exact recovery in the case when $p_n$ and $q_n$ are discrete distributions. Analogous to the scenario considered in \cite{abbe2014exact}, we now concentrate on the regime where the probability of having an edge scales as $\Theta\left(\frac{\log n}{n}\right)$. However, in addition to Bernoulli distributions, our framework accommodates distributions on a larger alphabet, denoted by the set $\{0, 1, \dots, L\}$ for $L \geq 1$. Thus, instead of simply observing the presence or absence of an edge, we may also observe the corresponding \emph{color} or \emph{weight} of the edge. We define the distributions $\{p_n, q_n\}$ as follows: For two vectors $\mathbf a = [a_1, a_2, \dots, a_L]$ and $\mathbf b = [b_1, b_2, \dots, b_L]$ in $\mathbb R_+^L$, define
\begin{align}
p_n(0) = 1 - \frac{u\log n}{n}, \quad  \text{and} \quad & p_n(\ell) = \frac{a_\ell\log n}{n}, \quad \forall 1 \leq \ell \leq L,\\
q_n(0) = 1 - \frac{v\log n}{n}, \quad \text{and} \quad & q_n(\ell) = \frac{b_\ell\log n}{n}, \quad \forall 1 \leq \ell \leq L,
\end{align}
where $u = \sum_{\ell=1}^L a_\ell$ and $v = \sum_{\ell=1}^L b_\ell$. We wish to determine a criterion in terms of $\mathbf a$ and $\mathbf b$ that describes when it is possible to to exactly determine the communities in this model. 

Our first result is the following theorem guaranteeing the success of the maximum likelihood estimator:
\begin{theorem}[Proof in Section~\ref{SecThmAchievable}]
\label{ThmAchievable}
Suppose
\begin{equation}
\label{EqnRenyiAB}
\sum_{\ell = 1}^L \left(\sqrt a_\ell - \sqrt b_\ell\right)^2 > 1.
\end{equation}
Then the maximum likelihood estimator recovers the communities exactly with probability converging to 1 as $n \rightarrow \infty$.
\end{theorem}
We note that the expression on the left-hand side of inequality~\eqref{EqnRenyiAB} is increasing in $L$, agreeing with the intuition that the exact recovery problem becomes easier when more edge colors are available: Given a graph with $L$ edge colors, we may always erase certain colors to obtain a new graph with $L' < L$ colors, and then apply a maximum likelihood estimator to the new graph. The probability of success of this estimator must be at least as large as the probability of success of a maximum likelihood estimator applied to the original graph; in particular, if
\begin{equation}
\label{EqnReduction}
\sum_{\ell=1}^{L'} \left(\sqrt{a_\ell} - \sqrt{b_\ell}\right)^2 > 1,
\end{equation}
implying that maximum likelihood succeeds with probability converging to 1 on the graph with $L'$ colors, the probability of success of maximum likelihood on the graph with $L$ colors must also converge to 1. Indeed, inequality~\eqref{EqnReduction} implies inequality~\eqref{EqnRenyiAB}, since $L' < L$. Similarly, we may check that by the Cauchy-Schwarz inequality, the following relation holds:
\begin{equation*}
\left(\sqrt{\sum_{\ell=1}^L a_\ell} - \sqrt{\sum_{\ell=1}^L b_\ell}\right)^2 \le \sum_{\ell=1}^L \left(\sqrt{a_\ell} - \sqrt{b_\ell}\right)^2.
\end{equation*}
This captures the fact that if the maximum likelihood estimator succeeds with probability converging to 1 on a graph with $L$ colors when we replace all occurring edges with a single color, then the maximum likelihood estimator on the original graph should also succeed with probability converging to 1.

\begin{remark}
Examining the proof of Theorem~\ref{ThmAchievable}, we may see that it is not necessary for the number of colors $L$ to be finite. Indeed, as long as we have
\begin{equation*}
\sum_{\ell=1}^\infty \left(\sqrt{a_\ell} - \sqrt{b_\ell}\right)^2 > 1,
\end{equation*}
in the infinite case, we will also have $\liminf_{n \rightarrow \infty} \frac{nI}{\log n} > 1$, implying the desired result.
\end{remark}

As will be seen in the proof of Theorem~\ref{ThmAchievable} below, we have the characterization
\begin{equation*}
I = \left( \sum_{\ell=1}^L \left(\sqrt a_\ell - \sqrt b_\ell\right)^2 \right) \frac{\log n}{n} + O\left(\frac{\log^2 n}{n^2}\right)
\end{equation*}
of the Renyi divergence. Hence, inequality~\eqref{EqnRenyiAB} governs whether $I < \frac{\log n}{n}$ or $I > \frac{\log n}{n}$, for large $n$. As will be illustrated in the computation appearing in the proof of Theorem~\ref{ThmAchievable}, the inequality $I > \frac{\log n}{n}$ implies that the right side of inequality~\eqref{eq: mallu1} tends to 0 as $n \rightarrow \infty$. On the other hand, the next theorem guarantees that if $I < \frac{\log n}{n}$, we have $\mprob(F)$ bounded away from 0. Hence, the success or failure of maximum likelihood occurs with respect to a sharp threshold that is encoded within the Renyi divergence. In the next theorem, we will make the additional assumption that
\begin{equation}
\label{EqnAbsCont}
a_\ell, b_\ell > 0, \qquad \forall 1 \le \ell \le L,
\end{equation}
meaning the probabilities of all $L$ colors are nonzero both within and between communities. 

\begin{theorem}[Proof in Section~\ref{SecThmImpossible}]
\label{ThmImpossible}
Suppose the condition~\eqref{EqnAbsCont} holds. If
\begin{equation*}
\sum_{\ell=1}^L \left(\sqrt a_\ell - \sqrt b_\ell\right)^2 < 1,
\end{equation*}
then for any $K \ge 2$ and for sufficiently large $n$, the maximum likelihood estimator fails with probability at least $\frac{1}{3}$.
\end{theorem}

Viewed from another angle, Theorems~\ref{ThmAchievable} and~\ref{ThmImpossible} imply that the quantity $\sum_{\ell=1}^L \left(\sqrt{a_\ell} - \sqrt{b_\ell}\right)^2$ determines a sharp threshold for when exact recovery is possible in the $K$-community weighted stochastic block model; when the quantity is larger than 1, the maximum likelihood estimator succeeds with probability converging to 1, whereas when the quantity is smaller than 1, the maximum likelihood estimator fails with probability bounded away from 0. Also note that the quantity is a sort of Hellinger distance between $\mathbf a$ and $\mathbf b$, although $\mathbf a$ and $\mathbf b$ need not be the probability mass functions of discrete distributions, since their components do not necessarily sum to 1.

\begin{remark}
The assumption~\eqref{EqnAbsCont} appears to be an undesirable artifact of the technique used to prove Theorem~\ref{ThmImpossible}, which involves bounding appropriate functions of the likelihood ratio between within-community and between-community distributions. However, it appears that a substantially different approach may be required to handle the case when assumption~\eqref{EqnAbsCont} does not necessarily hold. Furthermore, note that our argument also requires the likelihood ratio to be bounded by some constant $\cM$. Hence, although our impossibility proof continues to hold when $L$ is infinite, we will need to assume a bound of the form
\begin{equation*}
\sup_{\ell \ge 0} \left\{ \log \left( \frac{p_n(\ell)}{q_n(\ell)} \right)\right\} \le \cM
\end{equation*}
to establish the impossibility result when $L$ is infinite. (Such a bound clearly holds for finite values of $L$.)
\end{remark}

We also note that the results of Theorems~\ref{ThmAchievable} and~\ref{ThmImpossible} could be generalized further to include a mixture of discrete and continuous distributions. In other words, the distributions of $p_n(x)$ and $q_n(x)$ could follow arbitrary (discrete or continuous) distributions for the nonzero values, as long as
\begin{equation*}
p_n(0) = 1 - \frac{u \log n}{n}, \qquad \text{and} \qquad q_n(0) = 1 - \frac{v \log n}{n}.
\end{equation*}
This reflects the fact that the graph is still fairly sparse, with average degree scaling as $\Theta(\log n)$. However, whenever two nodes are connected by an edge, the distribution of the corresponding edge may follow a more general distribution.

\subsection{Censored block models and graphical channels}

We now discuss the relationship between our results and the notion of graphical channels introduced by Abbe and Montanari~\cite{AbbMon13}. Recall that a graphical channel takes as input a labeling of vertices on a graph, and each edge is encoded by a deterministic function of the adjacent vertices. The edges are then passed through a channel, and the output is observed.

Abbe et al.~\cite{AbbEtal14} analyze a specific instantiation of a discrete graphical channel known as the \emph{censored block model}. In this case, the node labelings are binary, and edges are encoded using the XOR operation on adjacent vertices. The channel is a discrete memoryless channel with output alphabet $\{\star, 0, 1\}$, and for fixed probabilities $p, q_1, q_2 \in [0,1]$, the transition matrix of the channel is given by
\begin{equation*}
\begin{blockarray}{cccc}
& \star & 0 & 1 \\
\begin{block}{c(ccc)}
0 & 1-p & p(1-q_1) & pq_1 \\
1 & 1-p & p(1-q_2) & pq_2 \\
\end{block}
\end{blockarray}.
\end{equation*}
In other words, an edge is replaced by $\star$ with probability $1-p$, and is otherwise flipped with probability $q_1$ or $1-q_2$, depending on whether the transmitted edge label is 0 or 1. Clearly, the observed graph may be viewed as a special case of the discrete model described in Section~\ref{SecDiscrete}, with $K = 2$ and $L = 2$, where $\star$ represents an empty edge and the two ``colors" are represented by 0 and 1. This leads to the following result, a corollary of Theorems~\ref{ThmAchievable} and~\ref{ThmImpossible}:
\begin{corollary}
\label{CorCensored}
In the censored block model, suppose
\begin{equation*}
\liminf_{n \rightarrow \infty} \left\{\frac{pn}{\log n} \left[\left(\sqrt{1-q_1} - \sqrt{1-q_2}\right)^2 + \left(\sqrt{q_1} - \sqrt{q_2}\right)^2\right]\right\} > 1.
\end{equation*}
Then the maximum likelihood estimator succeeds with probability converging to 1 as $n \rightarrow \infty$. On the other hand, if
\begin{equation*}
\limsup_{n \rightarrow \infty} \left\{\frac{pn}{\log n} \left[\left(\sqrt{1-q_1} - \sqrt{1-q_2}\right)^2 + \left(\sqrt{q_1} - \sqrt{q_2}\right)^2\right]\right\} < 1,
\end{equation*}
then the maximum likelihood estimator fails with probability bounded away from 0.
\end{corollary}
Sharp thresholds were derived for the censored block model by Abbe et al.~\cite{AbbEtal14} and Hajek et al.~\cite{HajEtal15} when $K = 2$ and $q_1 = 1 - q_2 = \epsilon$, in the cases where $\epsilon = \frac{1}{2}$ and $\epsilon \in [0,1]$, respectively. It is easy to check that their thresholds agree with ours. On the other hand, Corollary~\ref{CorCensored} does not require the graphical channel to flip edge labels with equal probability, and we may slightly relax the scaling requirement $p \asymp \frac{\log p}{n}$ in the statement of our corollary. Furthermore, the theorems in Section~\ref{SecDiscrete} clearly hold for more general graphical channels aside from the channel giving rise to the censored block model; we may have more than two labels for each node, corresponding to a larger codebook, and the output alphabet of the channel may be arbitrarily large. Translated into the language of graphical channels, our results from Section~\ref{SecDiscrete} show the following:
\begin{corollary}
\label{CorGraphicalChannel}
Consider a graphical channel, where node inputs are binary and edges are encoded using an XOR operation. The edges are passed through a discrete memoryless channel that maps each edge to a discrete label $\ell \in \{1, \dots, L\}$, with probability $\frac{a_\ell \log n}{n}$ for edges encoded with 0 and probability $\frac{b_\ell \log n}{n}$ for edges encoded with 1, and erases edges with probabilities $1 - \frac{\sum_{\ell = 1}^L a_\ell \log n}{n}$ and $1 - \frac{\sum_{\ell = 1}^L b_\ell \log n}{n}$, respectively. Let $I$ denote the Renyi entropy between the two output distributions. If $\liminf_{n \rightarrow \infty} \frac{nI}{\log n} > 1$, the maximum likelihood decoder succeeds with probability tending to 1. If $\limsup_{n \rightarrow \infty} \frac{nI}{\log n} < 1$, the maximum likelihood decoder fails with probability bounded away from 0.
\end{corollary}

As noted by Abbe and Sandon~\cite{AbbSan15} in a slightly different setting, the threshold for reliable communication in a graphical channel is governed by a different quantity from the mutual information between the input distribution and the output of the channel, which arises from the analysis of channel capacity in traditional channel coding theory. This is because the encoding of the graphical channel is already built into the stochastic block model framework, rather than being optimized by the user. It is interesting to observe that Renyi divergence and Hellinger distance are the information-theoretic quantities that determine the ``capacity" of graphical channels in the case of equal-sized communities.

\subsection{Thresholds for submatrix localization}\label{section: submatrix}

The stochastic block model framework described in this paper also has natural connections to the submatrix localization problem, in which our more general framework involving arbitrary (discrete or continuous) distributions is useful in deriving thresholds for exact recovery. The goal in submatrix localization is to partition the rows and columns of a random matrix $A \in \mathbb R^{n_L \times n_R}$ into disjoint subsets $\{C_1, \dots, C_K\}$ and $\{D_1, \dots, D_K\}$, where $n_L = \sum_{k=1}^K C_k$ and $n_R = \sum_{k=1}^K D_k$. For each $1 \le k \le K$, the entries $(i,j) \in C_k \times D_k$ are drawn i.i.d.\ from a distribution $G$ with mean $\mu_n > 0$, and all other entries in $A$ are drawn from the recentered distribution $G - \mu_n$.

Chen and Xu~\cite{CheXu14} derive impossibility and achievability results for submatrix localization when $|C_k| = K_L$ and $|D_k| = K_R$; i.e., the row and column subsets have equal size. Furthermore, the distribution $G$ is assumed to be sub-Gaussian with parameter 1. Chen and Xu~\cite{CheXu14} show that the maximum likelihood estimator succeeds with probability tending to 1 when
\begin{equation}
\label{EqnLocalizeUB}
\mu_n^2 \ge \frac{c_1 \log n}{\min\{K_L, K_R\}}.
\end{equation}
Furthermore, if $G \sim \cN(\mu_n, 1)$, the probability that maximum likelihood fails is bounded away from 0 when
\begin{equation}
\label{EqnLocalizeLB}
\mu_n^2 \le \frac{1}{12} \max\left\{\frac{\log(n_R - K_R)}{K_L}, \; \frac{\log(n_L - K_L)}{K_R}\right\}.
\end{equation}
Specializing to the case when $K_R = K_L = n$, inequalities~\eqref{EqnLocalizeUB} and~\eqref{EqnLocalizeLB} imply the existence of a threshold at $\mu^2 = \Theta\left(\frac{\log n}{n}\right)$, although the value of the constant has not been determined precisely. \\

When $K_R = K_L = n$, the results in Section~\ref{SecAchievable} may be applied to obtain sufficient conditions under which the maximum likelihood estimator succeeds for the submatrix localization problem with probability converging to 1. We have the following result, which follows directly from Corollary~\ref{CorRenyi} and the computation $I = \frac{\mu_n^2}{2}$ in the case when $G \sim \cN(\mu_n, 1)$:
\begin{corollary}
\label{CorLocalize}
Suppose $K_R = K_L = n$, and let $I$ denote the the Renyi divergence of order $\frac{1}{2}$ between the distributions $G$ and $G-\mu_n$. Suppose
\begin{equation}
\label{EqnGeneralCond}
\liminf_{n \rightarrow \infty} \frac{nI}{\log n} > 1.
\end{equation}
Then the maximum likelihood estimator succeeds with probability converging to 1. In particular, when $G \sim \cN(\mu_n, 1)$, maximum likelihood succeeds if
\begin{equation}
\label{EqnGaussianCond}
\liminf_{n \rightarrow \infty} \frac{n \mu_n^2}{\log n} > 4.
\end{equation}
\end{corollary}

In particular, note that the condition~\eqref{EqnGaussianCond} matches inequality~\eqref{EqnLocalizeUB}, with a value for the specific constant. Furthermore, the sufficient condition~\eqref{EqnGeneralCond} in Corollary~\ref{CorLocalize} may be of independent interest in obtaining thresholds for a general version of the submatrix localization problem, where the remaining entries in the martrix are drawn from a distribution $G'$ rather than a shifted version of $G$. For instance, if $G \sim \cN(\mu_n, \sigma_n^2)$ and $G' \sim \cN(\mu_n', \sigma_n'^2)$, the sufficient condition for exact recovery in Corollary~\ref{CorLocalize} becomes
\begin{equation*}
\liminf_{n \rightarrow \infty} \left\{\frac{(\mu_n - \mu_n')^2}{4\bar{\sigma}_n^2} + \log \left(\frac{\sigma_n'}{\sigma_n}\right) - 2 \log\left(\frac{\sigma_n'}{\bar{\sigma}_n}\right)\right\}  \frac{\log n}{n}> 1,
\end{equation*}
where $\bar{\sigma}_n^2 := \frac{\sigma_n^2 + \sigma_n'^2}{2}$. Although we do not yet have techniques for deriving impossibility results in the general submatrix localization setting, we conjecture that the upper bounds of Corollary~\ref{CorLocalize} based on the Renyi divergence may be tight here, as well.


\section{Proofs of theorems}
\label{SecProofs}

In this section, we outline the proofs of the main theorems. Detailed proofs of the more technical lemmas are contained in the appendix.

\subsection{Proof of Theorem~\ref{ThmRenyiML}}
\label{SecThmRenyiML}

We first show that the result holds when $p_n$ and $q_n$ are absolutely continuous with respect to each other. We provide a proof for the case when $p_n$ and $q_n$ are continuous distributions; the result for discrete distributions follows by replacing the integrals with summation signs. When $p_n$ and $q_n$ are not absolutely continuous with respect to each other, we establish the theorem for the two cases (continuous and discrete distributions) separately. \\

Define the function
\begin{equation*}
d_n(x) = \log \left(\frac{p_n(x)}{q_n(x)}\right).
\end{equation*}
We have the following lemma:

\begin{lemma}\label{lemma: aww}
Let the sets of vertices constituting the two communities be denoted by $A$ and $B$. If the maximum likelihood estimator does not coincide with the truth, then there exist $1 \leq k \leq \frac{n}{2}$ and sets $A_w \subset A$ and $B_w \subset  B$ such that $|A_w| = |B_w| = k$, and
\begin{equation}
\cS(A_w, \bar A_w) + \cS(B_w, \bar B_w) \leq \cS(A_w, \bar B_w) + \cS(\bar A_w, B_w).
\end{equation}
Here, $\bar A_w = A \setminus A_w$, $\bar B_w = B \setminus B_w$, and for disjoint sets of vertices $\hat A$ and $\hat B$,
$$\cS(\hat A, \hat B) := \sum_{i \in \hat A, j \in \hat B} d_n(w_{ij}).$$
\end{lemma}

\begin{proof}

Consider an assignment that is more likely than the maximum likelihood estimate. For this assignment, let $A_w$ and $B_w$ be the sets of misclassified nodes. Without loss of generality, we will assume that $k = |A_w| = |B_w| \leq n/2$. For disjoint sets of vertices $\hat A$ and $\hat B$, define $$p_n(\hat A, \hat B) = \prod_{i \in \hat A, j\in \hat B} p_n(w_{ij}),$$
and define $q_n(\hat A, \hat B)$ analogously. Since the new assignment is more likely that the truth, we must have
\begin{align*}
p_n(A_w, \bar A_w) \; p_n(B_w, \bar B_w) \; q_n(A_w, \bar B_w) \; q_n(\bar A_w, B_w) \leq q_n(A_w, \bar A_w) \; q_n(B_w, \bar B_w) \; p_n(A_w, \bar B_w) \; p_n(\bar A_w, B_w).
\end{align*}
Taking logarithms, this immediately implies that
\begin{align*}
\cS(A_w, \bar A_w) + \cS(B_w, \bar B_w) \leq \cS(A_w, \bar B_w) + \cS(\bar A_w, B_w),
\end{align*}
completing the proof.
\end{proof}

Let $F$ be the event that the maximum likelihood estimate does not coincide with the truth. For fixed sets $A_w$ and $B_w$ of size $k$, denote
\begin{equation*}
P_n^{(k)} = \mathbb P\left(\cS(A_w, \bar A_w) + \cS(B_w, \bar B_w) \leq \cS(A_w, \bar B_w) + \cS(\bar A_w, B_w) \right).
\end{equation*}
By Lemma~\ref{lemma: aww} and a union bound, we have
\begin{align}
\label{EqnPF}
\mathbb P(F) \leq \sum_{k=1}^{n/2} {n \choose k}^2 P_n^{(k)}.
\end{align}

Let $\{X_i\}_{i \geq 1}$ be a sequence of i.i.d.\ random variables distributed according to $p_n$, and let $\{Y_i\}_{i \geq 1}$ be a sequence of i.i.d.\ random variables distributed according to $q_n$. For natural number $N > 0$, define the expression
\begin{align}
\label{EqnT}
T(N, p_n, q_n, \epsilon) = \mathbb P \left(\sum_{i=1}^N \Big(d_n(Y_i)-d_n(X_i)\Big) \ge \epsilon \right).
\end{align}
Then
\begin{equation}
P_n^{(k)} = \mathbb P\left( \sum_{i=1}^{2k(n-k)} d_n(Y_i) - \sum_{i=1}^{2k(n-k)} d_n(X_i) \geq 0 \right) = T(2k(n-k), p_n, q_n, 0).
\end{equation}
Let $Z_i = d_n(Y_i) - d_n(X_i)$. The moment generating function of $Z_i$ is then given by
\begin{align*}
M(t) &= \mathbb E\left[e^{td_n(Y_i)}\right] \mathbb E\left[e^{-td_n(X_i)}\right] = \left(\int_{-\infty}^\infty \left(\frac{p_n(x)}{q_n(x)} \right)^t q_n(x) dx\right) \left(\int_{-\infty}^\infty \left(\frac{p_n(x)}{q_n(x)} \right)^{-t} p_n(x) dx\right).
\end{align*}
Let $t^\star$ be the the point where $M(t)$ is minimized for $t > 0$. We evaluate $t^\star$ by differentiating $M(t)$ and setting it to $0$, as follows:
\begin{align*}
M'(t) &=  \left(\int_{-\infty}^\infty \left(\frac{p_n(x)}{q_n(x)} \right)^t \log \frac{p_n(x)}{q_n(x)} q_n(x) dx\right) \left(\int_{-\infty}^\infty \left(\frac{p_n(x)}{q_n(x)} \right)^{-t} p_n(x) dx\right) \\
& \qquad + \left(\int_{-\infty}^\infty \left(\frac{p_n(x)}{q_n(x)} \right)^t q_n(x) dx\right) \left(\int_{-\infty}^\infty \left(\frac{p_n(x)}{q_n(x)} \right)^{-t} \log \frac{q_n(x)}{p_n(x)}p_n(x) dx\right).
\end{align*}
Note that if we substitute $t = 1/2$ in the above expression, we obtain
\begin{align*}
M'(1/2) &= \left(\int_{-\infty}^\infty \sqrt{p_n(x)q_n(x)}\log \frac{p_n(x)}{q_n(x)} dx\right) \left(\int_{-\infty}^\infty \sqrt{p_n(x)q_n(x)} dx\right) \\
& \qquad + \left(\int_{-\infty}^\infty \sqrt{p_n(x)q_n(x)} dx\right) \left(\int_{-\infty}^\infty \sqrt{p_n(x)q_n(x)} \log \frac{q_n(x)}{p_n(x)}dx\right)\\
&= 0.
\end{align*}
Since $M(t)$ is a convex function, we conclude that $t^\star = 1/2$. Substituting, we then obtain
\begin{equation*}
M(t^\star) = \left(\int_{-\infty}^\infty \sqrt{p_n(x)q_n(x)} dx\right)^2.
\end{equation*}
In particular,
\begin{equation*}
I = -\log M(t^\star) = -2\log \left(\int_{-\infty}^\infty \sqrt{p_n(x)q_n(x)} dx\right)
\end{equation*}
is the Renyi divergence defined in the statement of the theorem.

By a Chernoff bound on the sum $\sum_{i = 1}^{2k(n-k)} Z_i$, we have
\begin{align*}
P_n^{(k)} \leq \left(\inf_{t > 0} M(t) \right)^{2k(n-k)} = \left(\int_{-\infty}^\infty \sqrt{p_n(x)q_n(x)} dx\right)^{4k(n-k)} = \exp(-2k(n-k) I).
\end{align*}
Using ${n \choose k} \leq \left( \frac{ne}{k} \right)^k$ and substituting into inequality~\eqref{EqnPF}, we arrive at the bound
\begin{align}
\mathbb P(F) &\leq \sum_{k=1}^{n/2} \left( \frac{ne}{k} \right)^{2k}\exp(-2k(n-k) I) = \sum_{k=1}^{n/2} \exp \left( 2k \left( \log \frac{n}{k} + 1\right) -2k(n-k)I\right),
\end{align}
which is exactly inequality~\eqref{eq: mallu1}. As noted earlier, the proof for absolutely continuous discrete distributions follows exactly the same steps as above, and we will not repeat them here. \\

We now turn to the case where $p_n$ and $q_n$ are not necessarily absolutely continuous with respect to each other. 

\textbf{Case 1:} $p_n$ and $q_n$ are continuous distributions. Our strategy is to deliberately create a noisy version of the edges by adding a small Gaussian random variable to the existing edge weights, and then apply the maximum likelihood estimator to the new noisy graph. Naturally, this new estimator is worse than directly using a maximum likelihood estimator for the original distributions; however, the benefit of adding noise is that it makes the new distributions absolutely continuous with respect to each other. For some $\nu > 0$, we write $\hat p_n = p_n \star \cN(0,\nu^2)$ and $\hat q_n = q_n \star \cN(0, \nu^2)$, where $\star$ represents convolution. Let the Renyi divergence between $\hat p$ and $\hat q$ be denoted by $I_\nu$. Using the argument for absolutely continuous distributions, we conclude that
\begin{equation}
\mathbb P(F) \leq \sum_{k=1}^{n/2} \exp \left( 2k \left( \log \frac{n}{k} + 1\right) -2k(n-k)I_\nu\right).
\end{equation}
We claim that $\lim_{\nu \to 0} I_\nu = I$, which implies the desired result. From van Erven and Harremo\"{e}s~\cite{van2014renyi}, the Renyi divergence is uniformly continuous in $(P, Q)$, with respect to the total variation topology. Hence, it suffices to show that 
\begin{equation}
\lim_{\nu \to 0} ||\hat p_n - p_n||_1 = 0, \quad \text{and} \quad \lim_{\nu \to 0} ||\hat q_n - q_n||_1 = 0.
\end{equation}
The proof of the above fact is standard and may be found in Theorem 6.20 of Knapp~\cite{knapp2005basic} or the lecture notes~\cite{banuelos}.

\textbf{Case 2:} $p_n$ and $q_n$ are discrete distributions. Similar to the case of continuous distributions, we deliberately create a noisy graph and use the maximum likelihood estimator on this new graph. We fix an $\epsilon > 0$ and assume, without loss of generality, that $p_n(0), q_n(0) > 0$. We first replace every edge with weight 0 in the original graph by an edge with weight $i$, with probability $\frac{\epsilon}{2^i}$, for all $i \geq 1$. Thus, the new edge weight distributions are given by $\hat p_n$ and $\hat q_n$ where
\begin{align*}
\hat p_n(0) & = p_n(0)(1 - \epsilon), \text{ and } p_n(\ell) = p_n(\ell) + \frac{p_n(0)\epsilon}{2^\ell}, \text{ for } \ell \geq 1, \quad \text{and} \\
\hat q_n(0) & = q_n(0)(1 - \epsilon), \text{ and } q_n(\ell) = q_n(\ell) + \frac{q_n(0)\epsilon}{2^\ell}, \text{ for } \ell \geq 1.
\end{align*}
Since $\hat p_n$ and $\hat q_n$ are absolutely continuous with respect to each other, we have
\begin{equation}
\mathbb P(F) \leq \sum_{k=1}^{n/2} \exp \left( 2k \left( \log \frac{n}{k} + 1\right) -2k(n-k)I_\epsilon \right).
\end{equation}
where $I_\epsilon$ is the Renyi divergence between $\hat p_n$ and $\hat q_n$. It is easy to see that as $\epsilon \to 0$, we have $||\hat p_n - p_n||_1 \to 0$ and $||\hat q_n - q_n||_1 \to 0$. Again using the continuity of the Renyi divergence from van Erven and Harremo\"{e}s~\cite{van2014renyi}, we conclude that $$ \lim_{\epsilon \to 0} I_\epsilon = I,$$
which concludes the proof.


\subsection{Proof of Corollary~\ref{CorRenyi}}
\label{SecCorRenyi}

Note that for sufficiently large $n$, we have $I \ge (1+ \epsilon)\frac{\log n}{n}$, for some $\epsilon > 0$. Substituting into the bound~\eqref{eq: mallu1} of Theorem~\ref{ThmRenyiML}, we therefore have
\begin{align*}
\mathbb P(F) &\leq  \sum_{ k=1}^{n/2} \exp \left( 2k \left( \log \frac{n}{k} + 1\right) -2k(n-k)(1+ \epsilon)\frac{\log n}{n}\right)\\
&= \sum_{k=1}^{n/2} \exp\big(2k \left(\log n - \log k + 1 - (1-k/n)(1+ \epsilon)\log n \right) \big)\\
&= \sum_{k=1}^{n/2} \exp\big(2k\left(-\log k + 1 - (\epsilon - k/n - k\epsilon/n) \log n\right) \big)\\
&= \sum_{k=1}^{n/2} n^{-2k\epsilon} \exp \left( -2k\left(\log k - 1 - \frac{k\log n}{n} - \frac{k\epsilon\log n}{n} \right) \right).
\end{align*}
We break up the summation into two parts:
\begin{align}
\label{EqnScruffy}
& \sum_{k=1}^{n/2} n^{-2k\epsilon} \exp \left( -2k\left(\log k - 1 - \frac{k\log n}{n} - \frac{k\epsilon\log n}{n} \right) \right) \notag \\
&\qquad \qquad \qquad \qquad \qquad = \sum_{k=1}^2 n^{-2k\epsilon} \exp \left( -2k\left(\log k - 1 - \frac{k\log n}{n} - \frac{k\epsilon\log n}{n} \right) \right) \nonumber\\
& \qquad \qquad \qquad \qquad \qquad \qquad + \sum_{k = 3}^{n/2} n^{-2k\epsilon} \exp \left( -2k\left(\log k - 1 - \frac{k\log n}{n} - \frac{k\epsilon\log n}{n} \right) \right).
\end{align}
For $3 \le k \le \frac{n}{2}$, we have
\begin{equation}
\label{EqnBurfi}
\log k - \frac{k\log n}{n} - \frac{k \epsilon \log n}{n} = \log k - \frac{k(1+\epsilon)\log n}{n}  \geq \frac{\log k}{3}.
\end{equation}
This is the because the function $\frac{\log x}{x}$ is decreasing for $x \geq 3$, so we only need to verify that
\begin{equation*}
\frac{2}{3k}\log k \geq \frac{(1+\epsilon)\log n}{n}
\end{equation*}
holds at $k = n/2$. This is equivalent to checking that
\begin{equation*}
\frac{4}{3n}\log\left(\frac{n}{2}\right) = \frac{4}{3n}\log n - \frac{4\log 2}{3n} \ge \frac{(1+\epsilon)\log n}{n},
\end{equation*}
which indeed holds for sufficiently large $n$. Substituting the bound~\eqref{EqnBurfi} into inequality~\eqref{EqnScruffy}, we then obtain
\begin{multline}
\label{EqnLadoo}
\sum_{k=1}^{n/2} n^{-2k\epsilon} \exp \left( -2k\left(\log k - 1 - \frac{k\log n}{n} - \frac{k\epsilon\log n}{n} \right) \right) \\
\leq \sum_{k=1}^2 n^{-2k\epsilon} \exp \left( -2k\left(\log k - 1 - \frac{k\log n}{n} - \frac{k\epsilon\log n}{n} \right) \right) + \sum_{k = 3}^{n/2} n^{-2k\epsilon} \exp \left( -2k\left(- 1 + \frac{\log k}{3} \right) \right).
\end{multline}
The first term in inequality~\eqref{EqnLadoo} may be bounded as follows:
\begin{align*}
& \sum_{k=1}^2 n^{-2k\epsilon} \exp \left( -2k\left(\log k - 1 - \frac{k\log n}{n} - \frac{k\epsilon\log n}{n} \right) \right) \\
& \qquad \qquad = n^{-2\epsilon}\exp\left( 2 + \frac{2(1+\epsilon)\log n}{n}\right) + n^{-4\epsilon}\exp \left( -4\left(\log 2 - 1 - \frac{2\log n}{n} - \frac{2\epsilon\log n}{n} \right) \right) \\
& \qquad \qquad < Cn^{-2\epsilon},
\end{align*}
for a suitable constant $C$. For the second term in inequality~\eqref{EqnLadoo}, note that 
\begin{align*}
& \sum_{k = 3}^{n/2} n^{-2k\epsilon} \exp \left( -2k\left(- 1 + \frac{\log k}{3} \right) \right) \le n^{-6\epsilon} \sum_{k = 3}^{n/2} \exp\left( \frac{2k \log k}{3} \left(-1 + \frac{3}{\log k}\right)\right)\\
& \qquad \qquad \leq n^{-6\epsilon} \sum_{k=3}^{e^6} \exp\left( \frac{2k \log k}{3} \left(-1 + \frac{3}{\log k}\right)\right) + n^{-6\epsilon} \sum_{k = e^6 + 1}^\infty \exp\left( \frac{2k \log k}{3} \left(-1 + \frac{3}{\log k}\right)\right)\\
& \qquad \qquad \leq C_1 n^{-6\epsilon} + n^{-6\epsilon}\sum_{k = e^6 + 1}^\infty \exp \left(- \frac{k\log k}{3} \right)\\
& \qquad \qquad = O(n^{-6\epsilon}).
\end{align*}
Thus, we conclude that
\begin{align}
\mathbb P(F) \leq C_2 n^{-2\epsilon},
\end{align}
 for a suitable constant $C_2$, implying that $\mathbb P(F) \rightarrow 0$ as $n \rightarrow \infty$. This concludes the proof.


\subsection{Proof of Theorem~\ref{ThmRenyiK}}
\label{SecThmRenyiK}

We will follow the arguments used in the proof of Theorem 3.2 of Zhang and Zhou~\cite{zhangminimax}.

We label the nodes $\{1, 2, \dots, nK\}$. Without loss of generality, suppose community $i$ comprises the nodes $\{(i-1)n+1, \dots, in\}$, and denote the corresponding assignment mapping nodes to communities by $\sigma_0$. Let $A_{nK\times nK}$ be the adjacency matrix for the graph, where $A_{i,j} \in \{0, \dots, L\}$ is the color of edge $(i,j)$. Just as in the $K=2$ case, the maximum likelihood estimator for $K >2$ communities seeks the partition that minimizes the weight of cross-community edges (equivalently, maximizes the weight of within-community edges), where the weight of an edge with color \mbox{$\ell \in \{0, 1, \dots, L\}$} is given by
\begin{equation*}
w_\ell = \log \left\{\frac{p_n(\ell)}{q_n(\ell)}\right\}.
\end{equation*}
In other words, the maximum likelihood estimator $\hat \sigma$ satisfies
\begin{equation*}
\hat \sigma = \argmax_\sigma \sum_{i,j} w_\ell \cdot \mathbf 1\{ A_{i,j} = \ell \} \mathbf 1\{ \sigma(i) = \sigma(j) \} := \argmax_\sigma T(\sigma).
\end{equation*}
Note that the value of $T(\sigma)$ remains the same for permutations of $\sigma$. To be precise, let $\Delta$ be the set of all permutations from $\{1, \dots, K\}$ to $\{1, \dots, K\}$. For an assignment $\sigma$, denote 
$$\Gamma(\sigma) = \{\sigma' : \exists \delta \in \Delta \text{ s.t. } \sigma' = \delta \circ \sigma\}.$$
We may check that for all $\sigma' \in \Gamma(\sigma)$, we have $T(\sigma') = T(\sigma)$.  Thus, the maximum likelihood estimator finds the best equivalence class $\Gamma$ such that any $\sigma \in \Gamma$ achieves the maximum value of $T$.

From the equivalence class $\Gamma$, we pick a permutation $\sigma$ that is closest to $\sigma_0$ in terms of the Hamming distance. Let us denote this assignment by $\sigma(\Gamma)$; i.e.
$$\sigma(\Gamma) \in \text{argmin}_{\sigma \in \Gamma} d_H(\sigma, \sigma_0),$$
where $d_H(\sigma, \sigma_0)$ denotes the Hamming distance between $\sigma$ and $\sigma_0$. We now define
$$P_m := \mathbb P \left\{\exists\Gamma : d_H(\sigma(\Gamma), \sigma_0) = m \text{ and } T(\sigma(\Gamma))  \geq T(\sigma_0)\right\}.$$ 
Let $\mathbb P(F)$ be the probability that the maximum likelihood estimator fails. Clearly,
\begin{align}
\label{EqnSum}
\mathbb P(F) \leq \sum_{m=1}^{nK} P_m.
\end{align}
Furthermore, we have the inequality 
\begin{align*}
P_m \leq \left| \left\{\Gamma : d_H(\sigma(\Gamma), \sigma_0) = m \right\}\right| \cdot \max_{\{\sigma : d_H(\sigma, \sigma_0) = m\}} \mathbb P(T(\sigma) \geq T(\sigma_0)).
\end{align*}
We will bound each of the terms in the above product separately. For the first term, we use the following lemma:

\begin{lemma}[Proposition 5.2 of Zhang and Zhou~\cite{zhangminimax}]\label{lemma: gamma}
The cardinality of the equivalence classes $\Gamma$ such that $d_H(\sigma(\Gamma), \sigma_0) = m$ is bounded as follows:
$$\left| \left\{\Gamma : d_H(\sigma(\Gamma), \sigma_0) = m \right\}\right| \leq \min \left\{ \left( \frac{enK^2}{m}\right)^m, \; K^{nK}\right\}.$$
\end{lemma}
Suppose there exists an assignment $\sigma$ such that $d_H(\sigma, \sigma_0) = m$ and $T(\sigma) \geq T(\sigma_0)$. This is equivalent to
\begin{align*}
\sum_{i,j} w_\ell \cdot \mathbf 1\{ A_{i,j} = \ell \} \mathbf 1\{ \sigma(i) = \sigma(j) \} \ge \sum_{i,j} w_\ell \cdot \mathbf 1\{ A_{i,j} = \ell \} \mathbf 1\{ \sigma_0(i) = \sigma_0(j) \},
\end{align*}
or
\begin{align*}
\sum_{i,j:  \sigma(i) = \sigma(j), \sigma_0(i) \neq \sigma_0(j)} w_\ell \cdot \mathbf 1\{ A_{i,j} = \ell \}  \ge \sum_{i,j: \sigma(i) \neq \sigma(j), \sigma_0(i) = \sigma_0(j)} w_\ell \cdot \mathbf 1\{ A_{i,j} = \ell \}. 
\end{align*}
Denoting
$$\gamma = \Big|\{(i,j): \sigma(i) = \sigma(j), \sigma_0(i) \neq \sigma_0(j)\}\Big|, \quad \text{and} \quad \alpha =  \Big|\{(i,j): \sigma(i) \neq \sigma(j), \sigma_0(i) = \sigma_0(j)\}\Big|,$$
we then have the bound
\begin{align*}
\mathbb P(T(\sigma) \geq T(\sigma_0)) &= \mathbb P \left(\sum_{i=1}^\gamma d_n(Y_i) - \sum_{i=1}^{\alpha} d_n(X_i) \geq 0 \right) \notag \\
&\leq \inf_{t > 0} \left(\mathbb E e^{td_n(Y_1)}\right)^\gamma \left(\mathbb E e^{td_n(X_1)}\right)^\alpha \notag \\
&= \inf_{t > 0} \left( \sum_{\ell = 0}^L\left(\frac{p_n(\ell)}{q_n(\ell)}\right)^tq_n(\ell)\right) \cdot \left(\sum_{\ell = 0}^L\left(\frac{p_n(\ell)}{q_n(\ell)}\right)^{-t}p_n(\ell)\right),
\end{align*}
where as before, $d_n(\ell) = \log \left\{\frac{p_n(\ell)}{q_n(\ell)}\right\}$, and $X_i \sim p_n$ and $Y_i \sim q_n$, and we have used a Chernoff bound in the above inequality. Taking $t = 1/2$, we then arrive at
\begin{align}\label{eq: Tsigma1}
\mathbb P(T(\sigma) \geq T(\sigma_0)) &\leq e^{-(\gamma+\alpha) I} \leq e^{-\min(\gamma, \alpha)I},
\end{align}
where $I$ denotes the Renyi divergence of order $\frac{1}{2}$ between $p_n$ and $q_n$. 
We then use the following lemma:
\begin{lemma}[Lemma 5.3 of Zhang and Zhou~\cite{zhangminimax}]\label{lemma: bound alpha gamma}
For $0 < m < nK$, the minimum of $\alpha$ and $\gamma$ is bounded from below as follows:
\begin{equation*}
\min(\alpha, \gamma) \geq
\begin{cases}
nm-m^2, &\text{ if } m \leq \frac{n}{2}\\
\frac{2nm}{9}, &\text{ if } m > \frac{n}{2}.
\end{cases}
\end{equation*}
\end{lemma}

Substituting the bound from Lemma \ref{lemma: bound alpha gamma} into inequality \eqref{eq: Tsigma1}, we obtain the upper bound
\begin{equation}
\label{EqnTbound}
\mprob\left(T(\sigma) \geq T(\sigma_0)\right) \leq  
\begin{cases}
e^{(-nm+m^2)I}, &\text{ if } m \leq \frac{n}{2}\\
e^{-\frac{2mn}{9}I}, &\text{ if } m > \frac{n}{2}.
\end{cases}
\end{equation}
Finally, substituting the results of Lemma~\ref{lemma: gamma} and inequality~\eqref{EqnTbound} into inequality~\eqref{EqnSum}, we arrive at the bound~\eqref{EqnDucky}.

Note that we have
\begin{align*}
\exp{(-nm+m^2)I} < \exp\left(\frac{-2mn}{9}I\right),
\end{align*}
for $m < \frac{7n}{9}$. In particular, the bound~\eqref{EqnDucky} may be relaxed to obtain a bound of the form
\begin{equation}
\label{EqnGorilla}
\mprob(F) \le \sum_{m=1}^{m'} \min \left\{ \left( \frac{enK^2}{m}\right)^m, \; K^{nK}\right\} e^{(-nm + m^2)I} + \sum_{m= m'+1}^{nK} \min \left\{ \left( \frac{enK^2}{m}\right)^m, \; K^{nK}\right\} e^{-\frac{2mn}{9} I},
\end{equation}
for any $m' \leq \lfloor \frac{n}{2} \rfloor$. \\

We now verify the sufficiency of inequality~\eqref{EqnPiggy}. Suppose that for some $\epsilon > 0$ and for all sufficiently large $n$, we have
$$\frac{nI}{\log n} > 1+ \epsilon.$$ 
In particular, for $m' = \lfloor \frac{\epsilon n}{2} \rfloor$ and $m \in \{1, \dots, m'\}$, we have the bound
$$(n-m')I \ge n\left(1-\frac{\epsilon}{2}\right)I \geq n\left(1-\frac{\epsilon}{2}\right)(1+\epsilon)\frac{\log n}{n} > \left(1 + \frac{\epsilon}{3}\right)\log n,$$
for small enough $\epsilon$ and large enough $n$, implying that
\begin{align*}
P_m \leq \left( \frac{enK^2}{m} e^{-(n-m)I}\right)^m \leq\left( enK^2 e^{-(n-m')I}\right)^{m} \leq\left( eK^2 n^{-\epsilon/3}\right)^{m}.
\end{align*}
Thus,
\begin{align*}
\sum_{m = 1}^{m'} P_m &\leq \sum_{m=1}^{m'}\left( eK^2 n^{-\epsilon/3}\right)^{m} \\
& \leq \sum_{m=1}^{\infty}\left( eK^2 n^{-\epsilon/3}\right)^{m} \\
& \leq \left(eK^2 n^{-\epsilon/3}\right)\sum_{m=0}^{\infty}\left( eK^2 n^{-\epsilon/3}\right)^{m} \\
& \leq C_1n^{-\epsilon/3},
\end{align*}
where the last inequality follows because the geometric series converges for large enough $n$. 

For $m \in \{m'+1, \dots, nK\}$, we have the bound
\begin{align*}
P_m \leq \left( \frac{enK^2}{m} e^{-\frac{2nI}{9}}\right)^m \leq  \left( \frac{enK^2}{m'+1} e^{-\frac{2nI}{9}}\right)^m \le \left( \frac{2e}{\epsilon}K^2 e^{-\frac{2nI}{9}}\right)^m.
\end{align*}
Note that for large enough $n$, we also have
\begin{align*}
\frac{2nI}{9} > \frac{2n}{9}\frac{(1+\epsilon)\log n}{n} > \frac{2\log n}{9}.
\end{align*}
Hence,
\begin{equation*}
P_m \leq \left(\frac{2e}{\epsilon}K^2 n^{-\frac{2}{9}}\right)^m,
\end{equation*}
implying the bound
\begin{align*}
\sum_{m = m'+1}^{nK} P_m \leq \sum_{m = 1}^{\infty} \left(\frac{2e}{\epsilon}K^2 n^{-\frac{2}{9}}\right)^m \leq \left(\frac{2e}{\epsilon}K^2n^{-\frac{2}{9}}\right) \sum_{m = 0}^{\infty} \left(\frac{2e}{\epsilon} K^2 n^{-\frac{2}{9}}\right)^m \leq C_2n^{-\frac{2}{9}}.
\end{align*}
Therefore, using the decomposition~\eqref{EqnGorilla}, the total probability of failure is bounded by
\begin{align*}
\mathbb P(F) \leq C_1n^{-\epsilon/3} + C_2n^{-2/9},
\end{align*}
which converges to 0 as $n\to \infty$. This shows that the maximum likelihood estimator succeeds with probability tending to 1 as $n \rightarrow \infty$, as wanted.


\subsection{Proof of Theorem~\ref{ThmAchievable}}
\label{SecThmAchievable}

Note that in this setting, we have
\begin{align}
I &= -2 \log \left( \left(\sqrt{\left(1 - \frac{u\log n}{n}\right)\left(1 - \frac{v\log n}{n}\right)}\right) + \sum_{\ell=1}^L\frac{\sqrt{a_\ell b_\ell}\log n}{n}\right) \notag \\
&= -2 \log \left( \left(1 - \frac{u\log n}{2n} + O\left(\frac{\log^2 n}{n^2}\right)\right)\left(1 - \frac{v\log n}{2n} + O\left(\frac{\log^2 n}{n^2}\right)\right) + \sum_{\ell=1}^L  \frac{\sqrt{a_\ell b_\ell}\log n}{n}\right) \notag \\
&= -2 \log \left( 1 - \frac{u\log n}{2n} - \frac{v\log n}{2n} +  \sum_{\ell=1}^L  \frac{\sqrt{a_\ell b_\ell}\log n}{n} + O\left(\frac{\log^2 n}{n^2}\right)\right) \notag \\
&= -2  \left(- \frac{u\log n}{2n} - \frac{v\log n}{2n} +  \sum_{\ell=1}^L  \frac{\sqrt{a_\ell b_\ell}\log n}{n} + O\left(\frac{\log^2 n}{n^2}\right)\right) \notag \\
&= \frac{C \log n}{n} + O\left(\frac{\log^2 n}{n^2}\right), \label{eq: approx I}
\end{align}
where $C = \sum_{\ell=1}^L (\sqrt a_\ell - \sqrt b_\ell)^2$. 
In particular,
$$I = \left( \sum_{\ell=1}^L \left(\sqrt a_\ell - \sqrt b_\ell\right)^2 \right) \frac{\log n}{n} + O\left(\frac{\log^2 n}{n^2}\right).$$
Corollary~\ref{CorRenyi} (for $K = 2$ communities) and Theorem~\ref{ThmRenyiK} (for more than two communities) then imply the desired result.


\subsection{Proof of Theorem~\ref{ThmImpossible}}
\label{SecThmImpossible}

We will follow the proof strategy of Abbe et al.~\cite{abbe2014exact}. We will show that if 
$$\sum_{\ell=1}^L \left(\sqrt a_\ell - \sqrt b_\ell\right)^2 < 1,$$
there with a probability of at least $1/3$, we can find nodes $i \in A$ and $j \in B$ such that exchanging their community assignments has a larger likelihood than the ground truth. This would establish that the maximum likelihood estimator fails with probability at least $1/3$. Although we will establish the proof for the case of two communities, we note that the proof below trivially extends to $K > 2$ communities each of size $n$, simply by taking $A$ and $B$ to be any two fixed communities from the $K$ communities.  \\

Let $A = \{1, 2, \dots, n\}$ and $B = \{n+1, n+2, \dots, 2n\}$. For $i \neq j$, let $w_{ij} \in \{0, 1, \dots, L\}$ be the weight of the edge $(i,j)$. Just as in the case of unlabeled edges, maximizing the likelihood in the labeled case may be interpreted as finding the min-cut for the stochastic block model, where the weight of an edge with color $\ell \in \{0, \dots, L\}$ is $\log \left(\frac{p_n(\ell)}{q_n(\ell)}\right)$. For ease of notation, define the function $$d_n(\ell) = \log \left( \frac{p_n(\ell)}{q_n(\ell)} \right).$$  We may describe $d_n$ explicitly as
\begin{equation}
\label{EqnDefnDn}
d_n(0) = \log \left\{ \frac{1 - u\log n/n}{1 - v\log n/n}\right\}, \qquad d_n(\ell) = \log \left( \frac{a_\ell}{b_\ell} \right)\text{ for } 1 \leq \ell \leq K.
\end{equation}
Note that since $d_n(0) \to 0$ as $n \to \infty$, we may find a constant $\cM > 0$ that upper-bounds $d_n$ for all $n$. Thus, 
\begin{equation*}
\cM \geq \max_\ell d_n(\ell), \qquad \text{for all } n \text{ and all } 0 \leq \ell \leq L.
\end{equation*}
For any node $i$ and any subset of nodes $H$, denote 
$$\cS(i , H) = \sum_{j \in H, j \neq i} d_n(w_{ij}).$$
Using an argument along the lines of Lemma \ref{lemma: aww}, it is easy to check that if there exist nodes $i \in A$ and $j \in B$ such that
\begin{align}\label{eq: ml fail}
\cS(i, A\setminus \{i\}) + \cS(j, B\setminus \{j\}) < \cS(i, B\setminus \{j\}) + \cS(j, A\setminus \{i\}),
\end{align}
then the community assignment where $\sigma(i) = B$ and $\sigma(j) = A$ and every other assignment remains the same is more likely than the truth. Thus, the maximum likelihood estimator will fail if this happens. Define the following events:
\begin{align*}
F &= \text{ maximum likelihood fails},\\
F_A &= \exists i \in A ~:~ \cS(i, A \setminus \{i\}) < \cS(i, B) - \cM,\\
F_B &= \exists j \in B ~:~ \cS(j, B \setminus \{j\}) < \cS(i, A) - \cM.
\end{align*}
We have the following simple lemma:
\begin{lemma}\label{lemma: FA}
If $\mathbb P(F_A) \ge 2/3$, then $\mathbb P(F) \ge 1/3$.
\end{lemma}

\begin{proof}
By symmetry, we have $\mathbb P(F_B) \ge 2/3$, so by a union bound, $\mathbb P(F_A\cap F_B) \ge 1/3$. Thus, with probability at least $1/3$, there exist nodes $i\in A$ and $j \in B$ such that
\begin{align*}
\cS(i, A\setminus \{i\}) & < \cS(i, B) - \cM \leq \cS(i, B) - \cS(i,j) = \cS(i, B\setminus \{j\}), \qquad \text{and} \\
\cS(j, B\setminus \{j\}) & < \cS(j, A) - \cM \leq \cS(j, A) - \cS(i,j) = \cS(j, A\setminus \{j\}).
\end{align*}
This implies
\begin{equation*}
\cS(i, A\setminus \{i\})+\cS(j, B\setminus \{j\}) < \cS(i, B\setminus \{j\})+\cS(j, A\setminus \{j\}),
\end{equation*}
which from expression \eqref{eq: ml fail}, implies that the maximum likelihood estimator fails.
\end{proof}

We now define $\gamma(n)$ and $\delta(n)$ as follows:
\begin{align*}
\gamma(n) = (\log n)^{\log^{\frac{2}{3}} n}, \qquad \text{and} \qquad \delta(n) = \frac{\sqrt{\log n}}{\log \log n}.
\end{align*}
 Let $H$ be a fixed subset of $A$ of size $\frac{n}{\gamma(n)}$. We will take $\gamma(n) \asymp (\log n)^{\log^{\frac{2}{3}} n}$, such that $\frac{n}{\gamma(n)}$ is an integer. Define the event $\Delta$ as follows:
\begin{align*}
\Delta = \text{ for all nodes $i \in H$}, \quad  \cS(i, H) < \delta(n).
\end{align*}
We then have the following lemma:
\begin{lemma}[Proof in Appendix~\ref{AppLemDelta}] \label{lemma: delta}
$\mathbb P(\Delta) \geq \frac{9}{10}$.
\end{lemma}

Finally, define the events $F_H^{(i)}$ and $F_H$ as follows:
\begin{align*}
F_H^{(i)} &= \text{ node $i \in H$ satisfies } \cS(i, A \setminus H) + \delta(n) <  \cS(i, B) - \cM,\\
F_H &= \cup_{i \in H} F_H^{(i)},
\end{align*}
and define
\begin{align}
\label{EqnDefnRho}
\rho(n) = \mathbb P\Big(F_H^{(i)}\Big).
\end{align}
We have the following result:
\begin{lemma}\label{lemma: final blow}
If $\rho(n) > \frac{\gamma(n)\log 10}{n}$, then $\mathbb P(F) > 1/3$ for sufficiently large $n$.
\end{lemma}

\begin{proof}

We first show that $\mathbb P(F_H) > \frac{9}{10}$ for large enough $n$. 
Since the events $F_H^{(i)}$ are independent, we have
\begin{align*}
\mathbb P(F_H) = \mathbb P \left( \cup_{i \in H} F_H^{(i)} \right) = 1 - \mathbb P \left(  \cap_{i \in H} \left(F_H^{(i)}\right)^c \right) = 1 - (1-\rho(n))^{\frac{n}{\gamma(n)}}.
\end{align*}
Clearly, if $\rho(n)$ is not $o(1)$, then $\mathbb P(F)$ tends to 1 and we are done. If $\rho(n)$ is $o(1)$, then
\begin{align*}
\lim_{n \to \infty} (1-\rho(n))^{\frac{n}{\gamma(n)}} = \lim_{n \to \infty} (1-\rho(n))^{\frac{1}{\rho(n)} \frac{\rho(n)n}{\gamma(n)}} &= \lim_{n \to \infty} \exp \left(- \frac{\rho(n)n}{\gamma(n)} \right) < \frac{1}{10},
\end{align*}
where the last inequality used the fact that $\rho(n) > \frac{\gamma(n)\log 10}{n}$. Hence, $\mathbb P(F_H) > \frac{9}{10}$, as claimed.

Now note that $\Delta \cap F_H \subseteq F_A$. By Lemma~\ref{lemma: delta}, we also have $\mathbb P(\Delta) \ge \frac{9}{10}$. Hence,
$$\mathbb P(F_A) \geq \mathbb P(\Delta) + \mathbb P(F_H) - 1 \geq \frac{8}{10} > \frac{2}{3},$$
which combined with Lemma \ref{lemma: FA} implies the desired result.
\end{proof}

Let $\{X_i\}_{i \geq 1}$ be a sequence of i.i.d.\ random variables distributed according to $p_n$, and let $\{Y_i\}_{i \geq 1}$ be a sequence of i.i.d.\ random variables distributed according to $q_n$. From the definition~\eqref{EqnDefnRho} of $\rho(n)$, and using independence, we have
\begin{align}
\rho(n) &= \mathbb P \left( \sum_{i=1}^{n} d_n(Y_i) - \sum_{i=1}^{n-\frac{n}{\gamma(n)}} d_n(X_i) > \delta(n) + \cM \right) \notag \\
&\geq P \left( \sum_{i=1}^{n-\frac{n}{\gamma(n)}} d_n(Y_i) - \sum_{i=1}^{n-\frac{n}{\gamma(n)}} d_n(X_i) > \delta(n) + \cM - \hat \delta(n) \right)\times \mathbb P\left( \sum_{i= n-\frac{n}{\gamma(n)}+1}^n d_n(Y_i) \geq \hat \delta(n)\right), \label{eq: product}
\end{align}
for any $\hat \delta(n)$. We will choose a suitable $\hat \delta(n)$ so that
\begin{equation}
\label{EqnHari}
\mathbb P\left( \sum_{i= n-\frac{n}{\gamma(n)}+1}^n d_n(Y_i) \geq \hat \delta(n)\right) \longrightarrow 1.
\end{equation}
Note that $d_n(Y_i)$ is a random variable satisfying 
\begin{equation*}
\mathbb P \left( d_n(Y_i) = \log \left\{\frac{1-u\log n/n}{1 - v\log n/n}\right\}\right) = 1 - \frac{v\log n}{n}.
\end{equation*}
Thus, 
\begin{equation*}
\mathbb P \left(d_n(Y_i) = \log \left\{\frac{1-u\log n/n}{1 - v\log n/n}\right\}, \text{ for all } n-\frac{n}{\gamma(n)}-1 \leq i \leq n \right) = \left(1 - \frac{v\log n}{n}\right)^{\frac{n}{\gamma(n)}}.
\end{equation*}
We may check that
\begin{equation*}
\left(1 - \frac{v\log n}{n}\right)^{\frac{n}{\gamma(n)}} \longrightarrow 1, 
\end{equation*}
implying that
\begin{equation*}
\mathbb P \left(\sum_{i= n-\frac{n}{\gamma(n)}+1}^n d_n(Y_i) = \frac{n}{\gamma(n)} \cdot \log \left\{\frac{1-u\log n/n}{1-v\log n/n}\right\} \right) \longrightarrow 1.
\end{equation*}
Thus, equation~\eqref{EqnHari} holds with 
\begin{equation}
\hat\delta(n) = \Bigg|  \frac{n}{\gamma(n)} \cdot \log \left\{\frac{1-u\log n/n}{1-v\log n/n}\right\} \Bigg|.
\end{equation}
Since $$\hat \delta(n)  = O\left(\frac{\log n}{\gamma(n)}\right) = o(\sqrt{\log n}),$$ 
we have $\delta(n) + \cM - \hat\delta(n) = o(\sqrt{\log n})$.

Recall the definition of the function $T$ in equation~\eqref{EqnT}. We have the following technical lemma:

\begin{lemma}[Proof in Appendix~\ref{AppLemHorrible}] \label{lemma: horrible}
Let $\omega(n) = o(\sqrt{\log n})$ and $N(n) = n(1 + o(1))$. Then
$$-\log T\left(N(n), p_n, q_n, \omega(n)\right) \leq \left(\sum_{\ell=1}^L  \left(\sqrt{a_\ell} - \sqrt{b_\ell}\right)^2 \right)\log n + o(\log n).$$
\end{lemma}
Noting that
\begin{equation*}
T\left(n-\frac{n}{\gamma(n)}, p_n, q_n, \delta(n) + \cM -\hat \delta(n)\right) = \mathbb P \left( \sum_{i=1}^{n-\frac{n}{\gamma(n)}} d_n(Y_i) - \sum_{i=1}^{n-\frac{n}{\gamma(n)}} d_n(X_i) \geq \delta(n) + \cM - \hat \delta(n) \right),
\end{equation*}
and using Lemma~\ref{lemma: horrible}, we conclude that
\begin{equation}\label{eq: pumpkin}
 -\log T\left(n-\frac{n}{\gamma(n)}, p_n, q_n, \delta(n) + \cM -\hat \delta(n)\right) \leq  \left(\sum_{i=1}^L  \left(\sqrt{a_i} - \sqrt{b_i}\right)^2 \right)\log n + o(\log n).
\end{equation}

Substituting the bounds~\eqref{EqnHari} and~\eqref{eq: pumpkin} into equation \eqref{eq: product}, we then conclude that
\begin{equation*}
-\log \rho(n) \leq \left(\sum_{i=1}^L  \left(\sqrt{a_i} - \sqrt{b_i}\right)^2 \right)\log n + o(\log n).
\end{equation*}
In particular, when
$$\sum_{\ell=1}^L  \left(\sqrt{a_\ell} - \sqrt{b_\ell}\right)^2 < 1,$$
we have
$$-\log \rho(n) \leq \log n -\log \gamma(n) - \log\log 10,$$
for sufficiently large $n$. Lemma~\ref{lemma: final blow} then implies that maximum likelihood fails with probability at least $\frac{1}{3}$, completing the proof of the theorem.

\section{Discussion and open questions}
\label{SecDiscussion}

We have established thresholds for exact recovery in the framework of weighted stochastic block models, where edge weights may be drawn from arbitrary distributions. Whereas previous investigations had concentrated on the setting of unweighted edges, we show that the same techniques may be extended to the weighted case. Furthermore, the Renyi divergence of order $\frac{1}{2}$ between the distributions of edges coming from within-community and between-community connections arises as a fundamental quantity governing the hardness of the community estimation problem. \\

The conclusions of this paper leave open a number of open questions regarding phase transitions in general weighted stochastic block models. We conclude our paper by highlighting several interesting directions for future research.

\begin{itemize}

\item \textbf{Thresholds for exact recovery under continuous distributions.} Although the error bound for maximum likelihood derived in Theorem~\ref{ThmRenyiML} does not impose any conditions on the distributions $p_n$ and $q_n$, the proofs of the upper and lower bounds in Section~\ref{SecDiscrete} assume a specific setting involving discrete distributions with the same support. However, situations may arise where the observed edge weights are generated from continuous distributions. The submatrix localization problem in Section \ref{section: submatrix} provides one such example. It would be interesting to see if the Renyi divergence between $p_n$ and $q_n$ again plays a role in characterizing the threshold for exact recovery in the continuous case. However, a number of hurdles exist in extending our proof of impossibility to continuous distributions. Just as with discrete distributions, our proof technique does not allow for distributions that are not absolutely continuous with respect to each other. Furthermore, we have assumed the existence of a finite upper bound $\cM$ on the likelihood ratio between $p_n$ and $q_n$. Such a bound may not exist even for absolutely continuous distributions; for example, no such bound exists for $p_n = \cN(\mu_n, 1)$ and $q_n = \cN(0,1)$ in the submatrix localization problem. Finally, the emergence and relevance of the Renyi divergence term as a sharp threshold in this problem may be attributed in part to the specific regime we have considered, where the probabilities of connection scale according to $\Theta(\log n/n)$. Mossel et al.~\cite{MosEtal14} have shown that for Bernoulli distributions $p_n$ and $q_n$ in slightly denser regimes, where the probabilities scale according to $\Theta\left(\frac{\log^3 n}{n}\right)$, the threshold is no longer simply a function of the Renyi divergence.  

\item \textbf{General thresholds for weighted distributions.} Mossel et al.~\cite{MosEtal14} derive a very general theorem involving thresholds for the binary stochastic block model when $K = 2$. Defining
\begin{equation}
\label{EqnCamel}
P(n,p_n, q_n) = \mprob \left(\sum_{i=1}^n Y_i \ge  \sum_{i=1}^n X_i\right),
\end{equation}
where $X \sim p_n$ and $Y \sim q_n$, and $p_n$ and $q_n$ are Bernoulli distributions such that $p_n$ stochastically dominates $q_n$, Mossel et al.~\cite{MosEtal14} prove that exact recovery of the two communities is possible if and only if $P(n, p_n, q_n) = o\left(\frac{1}{n}\right)$. On the other hand, there exists an estimator for which the fraction of misclassified nodes converges to 0 if and only if $P(n, p_n, q_n) = o(1)$. It would be interesting to derive such a statement when $p_n$ and $q_n$ are general distributions, which could then be used to prove our results in Section~\ref{SecDiscrete} as a special case. Specifically, one might construct the analog of expression~\eqref{EqnCamel} to be
$$P(n, p_n, q_n) = \mathbb P \left(\sum_{i=1}^n d_n(Y_i) - \sum_{i=1}^n d_n(X_i) \geq 0 \right),$$
and conjecture analogous results about exact and partial recovery based on the rate at which $P(n, p_n, q_n)$ converges to 0.

\item \textbf{Efficient algorithms for exact recovery in weighted stochastic block models.} Hajek et al.~\cite{HajEtal14, HajEtal15} and Gao et al.~\cite{GaoEtal15} provide efficiently computable algorithms that achieve the threshold for exact recovery in the case of binary stochastic block models. Now that we have characterized the threshold for a more general class of weighted distributions, it would be interesting to see if similar efficient algorithms may be derived to obtain community assignments in the weighted case.

\end{itemize}


\appendix

\section{Proofs of technical lemmas}

In this section, we collect the proofs of the more technical lemmas used in proving the main results.

\subsection{Proof of Lemma~\ref{lemma: delta}}
\label{AppLemDelta}

Let $\Delta_i$ be the event  $\cS(i , H) < \delta(n).$ By a simple union bound calculation, we have
\begin{align*}
\mathbb P(\Delta) &=1- \mathbb P(\Delta^c) = 1 - \mathbb P \left( \cup_{i \in H} \Delta_i^c \right) \ge 1 - |H| \cdot \mathbb P(\Delta_i^c).
\end{align*}
We will show that $$|H| \cdot \mathbb P(\Delta_i^c) = o(1),$$ by showing that
$$\log |H| + \log \mathbb P(\Delta_i^c) \to -\infty,$$ as $n \to \infty$. 
Let the weights of the edges from $i$ to nodes within $H$ be the random variables $\{X_1, \dots, X_{|H| - 1}\}$. Note that the $X_i$'s are independent and identically distributed according to $p_n$. We have
\begin{align*}
\mathbb P(\Delta_i^c) &= \mathbb P\left( \cS(i, H) \geq \frac{\sqrt{\log n}}{\log \log n}\right) = \mathbb P\left( \sum_{j = 1}^ {|H|-1} d_n(X_i) \ge \frac{\sqrt{\log n}}{\log \log n}\right) \leq \inf_{t > 0} \left\{ \frac{\mathbb E\left[ e^{td_n(X_1)}\right]^{|H|-1}}{e^{\frac{t\sqrt{\log n}}{\log \log n}}}\right\},
\end{align*}
using a Chernoff bound in the last inequality. Thus, for $t > 0$, we have
\begin{align*}
\log |H| + \log \mathbb P(\Delta_i^c) &\leq \log \frac{n}{\gamma(n)} + \log \frac{\mathbb E\left[ e^{td_n(X_1)}\right]^{ \frac{n}{\gamma(n)}-1}}{e^{\frac{t\sqrt{\log n}}{\log \log n}}}\\
&= \log \frac{n}{\gamma(n)} + \left(\frac{n}{\gamma(n)}-1\right)\log \mathbb E\left[e^{td_n(X_1)} \right] - \frac{t\sqrt{\log n}}{\log \log n}.
\end{align*}
Picking $t = \sqrt{\log n}\log\log n$, the last expression simplifies to
\begin{align}\label{eq: master}
-\log \gamma(n) +  \left(\frac{n}{\gamma(n)}-1\right)\log \mathbb E\left[e^{\sqrt{\log n}(\log\log n) d_n(X_1)} \right].
\end{align}
We now analyze $\log \mathbb E\left[e^{\sqrt{\log n}(\log\log n) d_n(X_1)} \right]$ carefully. Note that
\begin{align*}
\log \mathbb E\left[e^{\sqrt{\log n}(\log\log n) d_n(X_1)} \right] &= \log \Bigg[\left( \frac{1-u\log n/n}{1 - v\log n/n}\right)^{\sqrt{\log n}\log \log n} \left(1-\frac{u\log n}{n}\right)\\
&\qquad \qquad \qquad \qquad \qquad + \sum_{\ell=1}^L \left(\frac{a_\ell}{b_\ell}\right)^{\sqrt{\log n}\log\log n} \frac{a_\ell\log n}{n}\Bigg]\\
&:= \log(1 + \mu_n + \nu_n),
\end{align*}
where
\begin{align*}
1 + \mu_n &= \left( \frac{1-u\log n/n}{1 - v\log n/n}\right)^{\sqrt{\log n}\log \log n} \left(1-\frac{u\log n}{n}\right), \quad \text{and} \\
\nu_n &= \sum_{\ell=1}^L \left(\frac{a_\ell}{b_\ell}\right)^{\sqrt{\log n}\log\log n} \left(\frac{a_\ell\log n}{n}\right).
\end{align*}

The following bound holds for $\nu_n$:
\begin{equation*}
\nu_n = \sum_{\ell=1}^L (\log n)^{\sqrt{\log n}\log \frac{a_\ell}{b_\ell}}\left(\frac{a_\ell \log n}{n}\right) \leq C_1 \frac{(\log n)^{C_2\sqrt{\log n}}}{n},
\end{equation*}
for suitable constants $C_1, C_2$. For $\mu_n$, we have
\begin{align*}
\mu_n &= \left( \frac{1-u\log n/n}{1 - v\log n/n}\right)^{\sqrt{\log n}\log \log n} \left(1-\frac{u\log n}{n}\right) - 1\\
&= \left(\left( \frac{1-u\log n/n}{1 - v\log n/n}\right)^{n/\log n}\right)^{\frac{(\log n)^{3/2}\log \log n}{n} } \left(1-\frac{u\log n}{n}\right) - 1.
\end{align*}
The term $\left( \frac{1-u\log n/n}{1 - v\log n/n}\right)^{n/\log n}$ tends to a constant, $\exp(v-u)$. Thus, for large enough $n$, we may find constants $0 <c_1 < c_2$ such that $\left( \frac{1-u\log n/n}{1 - v\log n/n}\right)^{n/\log n} \in (c_1, c_2)$. Using the Taylor series approximation of $c_i^x$ near 0, we have
$$c_i^{\frac{(\log n)^{3/2}\log \log n}{n}} = 1 + \frac{(\log n)^{3/2}\log \log n}{n}\log c_i + O\left(\left(\frac{(\log n)^{3/2}\log \log n}{n}\right)^2\right),$$
so
\begin{multline*}
c_i^{\frac{(\log n)^{3/2}\log \log n}{n}}\left(1-\frac{u\log n}{n}\right) - 1= \frac{(\log n)^{3/2} \log\log n}{n}\log c_i + O\left(\left(\frac{(\log n)^{3/2}\log \log n}{n}\right)^2\right) \\
- \frac{u\log n}{n} \left(1 + \frac{(\log n)^{3/2} \log \log n}{n} \log c_i\right).
\end{multline*}
Thus, for large enough $n$, there exists a constant $C_3$ that satisfies
\begin{equation*}
|\mu_n| \leq \frac{C_3\log^2 n}{n}.
\end{equation*}
Using the bound
\begin{align*}
\log (1 + \mu_n + \nu_n) \leq |\mu_n| +|\nu_n|,
\end{align*}
we conclude that
\begin{equation*}
\log \mathbb E\left[e^{\sqrt{\log n}(\log\log n) d_n(X_1)} \right] \leq C'_1 \frac{(\log n)^{C'_2\sqrt{\log n}}}{n},
\end{equation*}
for a suitable constants $C'_1$ and $C'_2$. Returning to the expression~\eqref{eq: master}, we conclude that
\begin{multline*}
-\log \gamma(n) +  \left(\frac{n}{\gamma(n)}-1\right)\log \mathbb E\left[e^{\sqrt{\log n} (\log\log n) d_n(X_1)} \right] \\
\le -\log \gamma(n) + \left(\frac{n}{\gamma(n)}-1\right) C'_1 \frac{(\log n)^{C'_2\sqrt{\log n}}}{n}.
\end{multline*}
Substituting $\gamma(n) = (\log n)^{\log^{\frac{2}{3}} n}$, we arrive at the upper bound
\begin{align*}
-\log^{\frac{2}{3}}n(\log\log n) + \left(\frac{n}{(\log n)^{\log^{\frac{2}{3}} n}} - 1\right)C'_1 \frac{(\log n)^{C'_2\sqrt{\log n}}}{n}.
\end{align*}
It is easy to check that as $n \rightarrow \infty$, we have
$$\left(\frac{n}{(\log n)^{\log^{\frac{2}{3}} n}} - 1\right)C'_1 \frac{(\log n)^{C'_2\sqrt{\log n}}}{n} \to 0,$$
and
$$-\log^{\frac{2}{3}}n(\log\log n) \to -\infty.$$
This concludes the proof.

\subsection{Proof of Lemma~\ref{lemma: horrible}}
\label{AppLemHorrible}

We will use the proof strategy found in Zhang and Zhou~\cite{zhangminimax}. Let 
$$Z = d_n(Y) - d_n(X),$$
where $X \sim p_n$ and $Y \sim q_n$. 
Let $M(t) = \mathbb E e^{tZ}$, and recall the following results from the proof of Theorem~\ref{ThmRenyiML}:
\begin{align*}
t^\star & = \arg\min_{t > 0} M(t) = \frac{1}{2}, \\
M(t^\star) & = \left(\sum_{\ell = 0}^L \sqrt{p_n(\ell)q_n(\ell)}\right)^2, \\
I & = -\log M(t^\star) = -2\log \left(\sum_{\ell = 0}^L \sqrt{p_n(\ell)q_n(\ell)}\right).
\end{align*}
Let $S_N = \sum_{i=1}^{N(n)} Z_i$, where the $Z_i$'s are i.i.d.\ and distributed according to $Z$, and denote the distribution of $Z$ by $p_Z$. Define $$\eta(n) = \log^{\frac{3}{4}} n.$$ Then
\begin{align}
\mathbb P\left(S_N \geq \omega(n)\right) &\geq \sum_{z: S_N \in [\omega(n), \eta(n))} \prod_{i=1}^{N(n)} p_Z(z_i) \notag \\
&\geq \frac{M^{N(n)}(t^\star)}{e^{t^\star \eta(n)}} \sum_{z: S_N \in [\omega(n), \eta(n))} \prod_{i=1}^{N(n)} \frac{e^{t^\star z_i}p_Z(z_i)}{M(t^\star)} \notag \\
&=\exp\left(-N(n)I - \frac{\eta(n)}{2}\right)\sum_{z: S_N \in [\omega(n), \eta(n))} \prod_{i=1}^{N(n)} \frac{e^{t^\star z_i}p_Z(z_i)}{M(t^\star)}, \label{eq: bumbum}
\end{align}
where the second inequality uses the fact that $e^{t^\star \eta(n)} \geq e^{t^\star \sum_i z_i}$ when $\sum_{i=1}^{N(n)} z_i < \eta(n)$.

Now denote $r(w) = \frac{e^{t^\star w}p_Z(w)}{M(t^\star)}$, and note that $r$ defines a probability distribution. Defining $W_1, W_2, \dots, W_n$ to be i.i.d.\ random variables with probability mass function $r(w)$, we then have
\begin{equation}
\label{eq: dundun}
\sum_{z: S_N \in [\omega(n), \eta(n))} \prod_{i=1}^{N(n)} \frac{e^{t^*z_i} p_Z(z_i)}{M(t^*)} = \mathbb P\left(\omega(n) \le \sum_{i=1}^{N(n)} W_i < \eta(n)\right).
\end{equation}
We also have the following concentration result:
\begin{lemma}[Proof in Appendix~\ref{AppLemCLT}]\label{lemma: clt}
Let $\{W_i\}_{i \geq 1}$ be i.i.d.\ random variables distributed as $r(w)$. Then 
\begin{equation*}
\frac{\sum_{i=1}^n W_i}{\sqrt{\log n}} \stackrel{d}{\longrightarrow} \cN(0, \nu^2),
\end{equation*}
as $n \rightarrow \infty$, where $\nu > 0$ is a constant.
\end{lemma}

By Lemma~\ref{lemma: clt}, it follows that
$$\frac{1}{\sqrt {\log N(n)}}\sum_{i=1}^{N(n)} W_i \stackrel{d}\to \cN(0, \nu^2),$$
for some constant $\nu > 0$. Furthermore, by our choices of $\omega(n)$, $N(n)$, and $\eta(n)$, we have 
\begin{align*}
\frac{\omega(n)}{\sqrt{\log N(n)}} \to 0, \qquad \text{and} \qquad \frac{\eta(n)}{\sqrt{\log N(n)}} \to +\infty.
\end{align*}
Thus, 
\begin{equation*}
 \mathbb P\left(\frac{\omega(n)}{\sqrt{\log N(n)}} \leq \frac{1}{\sqrt {\log N(n)}}\sum_{i=1}^{N(n)} W_i < \frac{\eta(n)}{\sqrt{\log N(n)}} \right) \to 1/2,
\end{equation*}
implying that the left-hand probability expression becomes larger that $1/4$ for all large enough $n$.
Combining this with the bounds~\eqref{eq: bumbum} and~\eqref{eq: dundun}, we then obtain
\begin{align*}
\mathbb P(S_N \geq \omega(n)) \geq \exp \left(-N(n)I - \frac{\log^{\frac{3}{4}} n}{2} - \log 4 \right).
\end{align*}

Now recall the computation in equation~\eqref{eq: approx I}. Using $N = n(1+o(1))$, we arrive at
\begin{align*}
-\log T\left(N(n), p_n, q_n, \omega(n)\right) &= -\log \mathbb P(S_N \geq \omega(n)) \leq \left(\sum_{\ell=1}^L \left(\sqrt a_\ell - \sqrt b_\ell \right)^2\right)  \log n+ o(\log n).
\end{align*}
This concludes the proof.

\subsection{Proof of Lemma~\ref{lemma: clt}}
\label{AppLemCLT}

We show that the moment generating function of $\frac{\sum_{i=1}^n W_i}{\sqrt{\log n}}$ converges to that of a normal random variable. By a simple computation, we may check that $r$ is a sum of delta distributions with mass
\begin{equation}
\label{EqnDefnZeta}
\zeta\left(\log \frac{p_n(y)}{q_n(y)} - \log \frac{p_n(x)}{q_n(x)}\right) = \frac{\sqrt{p_n(x)q_n(x)p_n(y)q_n(y)}}{\left(\sum_{\ell = 0}^L \sqrt{p_n(\ell)q_n(\ell)}\right)^2},
\end{equation}
at the point $\log \frac{p_n(y)}{q_n(y)} - \log \frac{p_n(x)}{q_n(x)}$, for all $0 \le x, y \le L$. Note that the right-hand side of equation~\eqref{EqnDefnZeta} is symmetric with respect to $x$ and $y$, implying that $r$ is a symmetric distribution. For $x, y\neq 0$, we then have
\begin{align*}
\zeta\left(\log \frac{p_n(y)}{q_n(y)} - \log \frac{p_n(x)}{q_n(x)}\right) &= \frac{\sqrt{a_x b_x a_y b_y}}{\left(\sum_{\ell = 0}^L \sqrt{p_n(\ell)q_n(\ell)}\right)^2} \cdot \frac{\log^2 n}{n^2} = O\left(\frac{\log^2 n}{n^2}\right).
\end{align*}
For $x = 0$ and $y \neq 0$ (and by symmetry, for $y = 0$ and $x \neq 0$), we have
\begin{align*}
\zeta\left(\log \frac{p(y)}{q(y)} - \log \frac{p(0)}{q(0)}\right) &= \frac{\sqrt{(1-u\log n/n)(1-v\log n/n)a_y b_y}}{\left(\sum_{\ell = 0}^L \sqrt{p_n(\ell)q_n(\ell)}\right)^2} \cdot \frac{\log n}{n}\\
&= \frac{C_y\log n}{n} + O\left(\frac{\log^2 n}{n^2}\right),
\end{align*}
for a suitable constant $C_y > 0$. Hence,
\begin{equation*}
r(0) = 1 - \frac{C_0\log n}{n} + O\left(\frac{\log^2 n}{n^2}\right),
\end{equation*}
for some constant $C_0 > 0$.

We now examine the range of $W \stackrel{d}{=} W_i$, which we denote by the set $\cW$. Note that the range is finite, since $W$ can only take values from set $\left\{\log \left(\frac{p_n(y)q_n(x)}{q_n(y)p_n(x)}\right): 0 \leq x, y \leq K \right\}$. Also note that the range depends on $n$, since the ratio $\log \left(\frac{p_n(0)}{q_n(0)}\right)$ changes with $n$. However, since $\log \left(\frac{p_n(0)}{q_n(0)}\right) = O\left(\frac{\log n}{n}\right)$, this dependence may only perturb the range by $O\left(\frac{\log n}{n}\right)$. Thus, we may fix constants $\{0, \pm w_1, \dots, \pm w_R\}$ such that the range of $W$ is given by
\begin{equation*}
\cW = \{0, \pm \hat w_1, \dots, \pm \hat w_R\} \text{ where } \hat w_i = w_i + O\left(\frac{\log n}{n}\right), \text{ for } 1 \leq i \leq R.
\end{equation*}
Since $W$ is a symmetric random variable, it is easy to see that its moment generating function is given by
\begin{equation}
\label{EqnMGF}
\mathbb E e^{tW} = 1 + \sum_{j = 1}^R r(\hat w_j)\left(e^{t\hat w_j/2} - e^{-t\hat w_j/2}\right)^2,
\end{equation}
using the fact that $r(0) = 1 - \sum_{j=1}^R 2r(\hat w_j)$. As noted above, for certain nonzero $\hat w \in \cW$, we have $r(\pm \hat w) = \Theta\left(\frac{\log n}{n}\right)$; whereas for other values, we have $r(\hat w) = O\left( \frac{\log^2 n}{n^2}\right)$. Without loss of generality, let $r(\hat w_j)$, for $1 \leq j \leq N$, be $\Theta\left(\frac{\log n}{n} \right)$, and let $r(\hat w_j)$, for $ N+1 \leq j \leq R$, be $O\left(\frac{\log^2 n}{n^2}\right)$. We then write $r(\hat w_j) = \frac{C_j \log n}{n} + O\left(\frac{\log^2 n}{n^2}\right)$, for $1 \leq j \leq N$. Using the expression~\eqref{EqnMGF}, the moment generating function of $W$ is then given by
\begin{multline*}
\mathbb E e^{tW} = 1 + \sum_{1 \leq j \leq N} \left(\frac{C_j \log n}{n} + O\left(\frac{\log^2 n}{n^2}\right)\right)\left(e^{t\hat w_j/2} - e^{-t\hat w_j/2}\right)^2 \\
+ \sum_{N+1 \leq j \leq R} O\left(\frac{\log^2 n}{n^2} \right) \left(e^{t\hat w_j/2} - e^{-t\hat w_j/2}\right)^2.
\end{multline*}
Substituting $\frac{t}{\sqrt{\log n}}$ in place of $t$ and using the approximation $a^{x/2} - a^{-x/2} = x\log a + O(x^2\log^2 a)$ for $x = o(1)$, we arrive at
\begin{align*}
\mathbb E e^{tW/\sqrt{ \log n}} &=  1 + \sum_{1 \leq j \leq N} \left(\frac{C_j \log n}{n} + O\left(\frac{\log^2 n}{n^2}\right)\right) \left(\frac{t\hat w_j}{\sqrt {\log n}} + O\left(\frac{1}{\log n}\right)\right)^2\\
& \qquad + \sum_{N+1 \leq j \leq R} O\left(\frac{\log^2 n}{n^2} \right)\left(\frac{t\hat w_j}{\sqrt {\log n}} + O\left(\frac{1}{\log n}\right)\right)^2\\
%
%
%
&= 1 + \frac{Ct^2}{n} + o\left(\frac{1}{n}\right),
\end{align*}
for a suitable constant $C$, where the second equality uses the fact that $\hat w_j = w_j + O\left(\frac{\log n}{n}\right)$, for all $1 \leq j \leq R$. Hence, the moment generating function of $\frac{\sum_{i=1}^n W_i}{\sqrt{\log n}}$ is given by
\begin{align*}
\left(\mathbb E e^{tW/\sqrt{ \log n}}\right)^n &= \left(1 + \frac{Ct^2}{n} + o\left(\frac{1}{n}\right)\right)^n \longrightarrow e^{Ct^2},
\end{align*}
which is the moment generating function of $\cN(0, 2C)$. This completes the proof.

\bibliography{refs}

\end{document}